\documentclass[sigconf]{aamas}  

\usepackage{booktabs}
\usepackage{amsmath,amsfonts,mathrsfs,amssymb,bm}
\usepackage{graphicx}
\newfloat{model}{thp}{lop}
\floatname{model}{Model}

\usepackage{todonotes}
\usepackage{xcolor,colortbl}
\usepackage[ruled,linesnumbered,noresetcount,vlined]{algorithm2e}
\SetKwFor{Repeat}{repeat}{:}{endw}

\usepackage[nolist,nohyperlinks]{acronym}

\usepackage{hyperref} 
\hypersetup{
colorlinks=true, 
breaklinks=true,
urlcolor= blue, 
linkcolor= blue,
citecolor=red, 
}
\usepackage{float}
\floatstyle{ruled}
\newfloat{model}{thp}{lop}
\floatname{model}{Model}
\usepackage{multicol}
\usepackage{multirow}

\usepackage{booktabs}
\usepackage{longtable}
\usepackage{xcolor}
\usepackage{enumitem}
\usepackage[export]{adjustbox}
\usepackage{graphicx,subcaption,cleveref}

\newcommand{\bitemize}{\begin{list}{$\bullet$}{\topsep=1pt \parsep=0pt \itemsep=1pt \leftmargin=1em }} 
\newcommand{\eitemize}{\end{list}}
\newcommand{\benumerate}{\hspace{-0.5in}\begin{enumerate}\topsep=0pt \parsep=0pt \itemsep=-3pt} 
\newcommand{\eenumerate}{\end{enumerate}}
\newcommand{\beitemize}{\begin{list}{$\bullet$}{\topsep=1.5pt \parsep=0pt \itemsep=1pt \leftmargin=1em }} 
\newcommand{\enitemize}{\end{list}}

\newcommand{\D}{{\mathscr{D}}}	
\newcommand{\R}{{\mathscr{R}}}	

\newcommand{\cM}{\mathcal{M}}

\newcommand{\cO}{\mathcal{O}}
\newcommand{\cP}{\mathcal{P}}	%

\newcommand{\cZ}{\mathcal{Z}}

\newcommand{\bd}{\bm{d}}

\newcommand{\bx}{\bm{x}}
\newcommand{\by}{\bm{y}}

\newcommand{\EE}{\mathbb{E}}
\newcommand{\RR}{\mathbb{R}}

\newcommand{\Lap}{{\text{Lap}}}

\usepackage{flushend}

\begin{document}
\sloppy\allowdisplaybreaks
\title{PPSM: A Privacy-Preserving Stackelberg Mechanism}
\subtitle{Privacy Guarantees for the Coordination of Sequential Electricity and Gas Markets}


\author{Ferdinando Fioretto}
\authornote{Authors names listed alphabetically. All authors have equal contributions.}
\affiliation{%
  \institution{Georgia  Institute of Technology
   \city{Atlanta} 
   \state{GA} 
  }
  \institution{Syracuse University
    \city{Syracuse} 
   \state{NY} 
  }
}
\email{ffiorett@syr.edu}
\author{Lesia Mitridati}
\affiliation{%
 \institution{Georgia Institute of Technology
    \city{Atlanta} 
    \state{GA} 
  }
}
\email{lmitridati3@gatech.edu}
\author{Pascal Van Hentenryck}
\affiliation{%
 \institution{Georgia Institute of Technology
    \city{Atlanta} 
    \state{GA} 
  }
}
\email{pvh@isye.gatech.edu}

\renewcommand{\shortauthors}{F. Fioretto et al.}

\begin{abstract}  
This paper introduces a \emph{differentially private} mechanism to protect the information exchanged during the coordination of the sequential market clearing of electricity and natural gas systems. 
The coordination between these \textit{sequential} and \textit{interdependent} markets represents a classic Stackelberg game and relies on the exchange of \emph{sensitive information} between the system agents, including the supply and demand bids in each market or the characteristics of the systems. The paper is motivated by the observation that traditional differential privacy mechanisms are unsuitable for the problem of interest: The perturbation introduced by these mechanisms fundamentally changes the underlying optimization problem and even leads to unsatisfiable instances. 
To remedy such limitation, the paper introduces the \emph{Privacy-Preserving Stackelberg Mechanism} (PPSM), a framework that enforces the notions of \textit{consistency} and \textit{fidelity} of the privacy-preserving information to the original problem objective. The PPSM has strong properties: It complies with the notion of differential privacy and ensures that the outcomes of the privacy-preserving coordination mechanisms are close-to-optimality for each agent. The fidelity property is analyzed by providing theoretical guarantees on the \textit{cost of privacy} of PPSM and 
experimental results on several gas and electricity market benchmarks based on a real case study demonstrate the effectiveness of the approach.
\end{abstract}

\keywords{Differential Privacy; Stackelberg Games; Electricity and Gas Markets}  

\maketitle
\section{Introduction}

The liberalization of the energy sectors along with the growing efforts to 
achieve a sustainable energy future have lead to an increased competition 
and the decentralization of decision-making in energy systems, such 
as electricity, heat, and natural gas systems \cite{pinson2017towards}. As a result, the coordination among multiple rational agents has 
become central to achieve efficient and cost-effective operations of 
the overall energy system. In particular, the growing share of \emph{flexible gas-fired power plants} (GFPPs) at the interface between electricity and natural gas systems creates strong interdependencies between their respective markets \cite{ordoudis2018market,alabdulwahab2015coordination}. 

The coordination of \textit{sequential} and \textit{interdependent} agents in the energy system, such as electricity and district heating markets \cite{mitridati2016optimal,mitridati2019bid}, electricity transmission and distribution system operators \cite{yuan2017hierarchical,le2019game,hassan2018energy}, or aggregators and consumers \cite{zhang2016real,momber2015retail}, has traditionally been modeled as a \textit{Stackelberg game} \cite{simaan1973stackelberg}. 
While this coordination approach brings substantial economical benefits to the power system, it also requires the exchange of proprietary information between the agents in order to achieve an optimal strategy. 
Relevant data may represent the costs of producers or the loads of consumers, or technical characteristics of the network. Such information is however considered sensitive as it can provide a competitive advantage over other strategic agents in the energy system, and may incur financial losses. Additionally, a breach of privacy on these parameters may benefit an attacker. For instance, \cite{maharjan2013dependable} analyzed the impact of an attacker who manipulates the information exchange between the aggregator and consumers in the power system and showed the loss in social welfare for both agents.

To address this issue, several privacy-preserving framework have been proposed. In particular, \emph{Differential Privacy} (DP) \cite{dwork:06} captures a desirable privacy property of computations over a dataset. It allows to measure and bound the risk associated with an individual participation to an analysis task. Differential privacy algorithms rely on the injection of carefully calibrated noise to the output of a computation. They can thus be used to \emph{obfuscate} the sensitive data exchanged by the system agents in the market. However, as the paper observes in Section \ref{sec:experiments}, when these privacy-preserving mechanisms are used as input to complex optimization problems, as in the case of \emph{Stackelberg games}, they may produce results that are fundamentally different from those obtained on the original data: They often transform the nature of the underlying optimization problem of the agents, and even lead to \emph{severe feasibility issues}.

As a consequence, despite its strong theoretical foundations, industrial adoption of differential privacy has remained limited. Large-scale deployments were carried out by large data owners, such as Google \cite{fanti:16} and Apple \cite{apple}. These applications, however, do not involve data used for solving complex optimization problems, but rather for evaluating a pre-defined set of queries, e.g., the count of individuals satisfying specific criteria for statistical analysis. 

This paper is motivated by the desire of solving complex coordination problems arising in the sequential and interdependent markets operating the electricity and gas systems, while protecting the privacy of the demand information exchanged in these markets. 

\paragraph{\textbf{Contributions}}
This paper makes several contributions to the state of the art.
Firstly, it introduces the \emph{Privacy-Preserving Stackelberg Mechanism} (PPSM) for the coordination of electricity and natural gas market agents. The approach relies on hierarchical optimization to model the coordination problem and on the notion of differential privacy to protect the exchange of proprietary information between the agents. 
Secondly, it introduces two optimization-based approaches that allow the PPSM to achieve \textit{fidelity} of the privacy-preserving exchanged information with respect to the original problem. These approaches rely on \emph{fidelity constraints} on the primal and dual decision variables of the agents. 
Thirdly, it discusses theoretical guarantees on the PPSM \textit{cost of privacy}. 
Finally, the approach is validated on a real test case for the Northeastern United States \cite{byeon2019unit} and show to bring up to two orders of magnitude error reduction over competitor privacy-preserving mechanisms.

\section{Preliminaries}

\subsection{
Coordination of Sequential Markets}

The strategies of two sequential and interdependent agents, such as market operators in energy systems, represent a classic \emph{Stackelberg game} \cite{simaan1973stackelberg}. In this game-theoretical framework, the \emph{leader}, i.e. the first market clearing, optimizes its decisions while anticipating on the reaction of the \emph{follower}, i.e. the second market clearing.

The leader actions impact the reaction of the follower, through its feasible decision space, which in turn impacts the leader objective and feasible space. As a result,
the leader strategy in Stackelberg games can be modeled as a hierarchical optimization problem \cite{gabriel2012complementarity}:
\begin{subequations} \label{eq:hierarchical}
\begin{alignat}{2}
  \!\!\!\!S=\min_{\bx^L,\bx^F,\by^F} &
  	\;\; \cO^L 
  		\left(\bx^L,\bx^F,\by^F,D^L\right) && \quad 
  		\label{eq:hierarchical0}\\
  \text{s.t.}\;\;& 
  \bx^L  \in \cZ^L \left(\bx^F,\by^F,D^L\right) && \quad 
  \label{eq:hierarchical1}\\
  (\bx^F,\by^F) &= \text{primal, dual sol.~of } &&
   \min_{\bx^F} \; \cO^F \left(\bx^F,D^F\right) 
   \label{eq:hierarchical2}\\ 
	& \quad && \text{s.t. } \; \bx^F  \in \cZ^F \left(\bx^L,D^F\right), 
	\label{eq:hierarchical3}
\end{alignat}
\end{subequations}
where $\bx^L$ represents the vector of decision variables of the leader, and $\bx^F$ and $\by^F$ the vectors of primal and dual decision variables of the follower. Additionally, $D^L$ and $D^F$ denote the vectors of parameters in the leader and follower problems, respectively. The follower parameters may be differentiated between \textit{public} and \textit{sensitive}, represented by the vectors $D^P$ and $D^S$, respectively, such that $D^F = \left[ D^P, D^S \right]$.
Taking these parameters as input, the upper-level problem minimizes the leader objective cost \eqref{eq:hierarchical0}, constrained by the feasible space of the leader decisions \eqref{eq:hierarchical1}, and the reaction of the follower in the lower-level problem \eqref{eq:hierarchical2} and \eqref{eq:hierarchical3}. 
The lower-level problem minimizes the follower objective cost \eqref{eq:hierarchical2}, constrained by the feasible space of the follower decisions \eqref{eq:hierarchical3}.

\subsubsection*{\textbf{Coordination Variables}}
Note that, the leader decision variables appear as fixed parameters in the expression of the follower feasible space $\cZ^F$. In return, the lower-level problem provides feedback on the follower decision variables to the upper-level problem through its objective function $\cO^L$ and feasible space $\cZ^L$.
The variables shared between the follower and the leader are called \textit{coordination variables}. These variables may be classified into two categories, i.e., \textit{primal} and \textit{dual} variables. In the context of the coordination between two energy markets, primal coordination variables may represent the \textit{quantity} of energy dispatched, while dual coordination variables may represent the market-clearing \textit{prices}. 

The remainder of the paper details the coordination between natural gas and electricity markets to motivate and illustrate the privacy-preserving mechanisms proposed. However, the methods developed may apply more generally to any type of coordination mechanism between sequential and interdependent agents.


\subsection{Privacy Goal}
\label{sec:privacy_goal}

The coordination mechanisms \eqref{eq:hierarchical} requires information exchange between the leader and the follower. 
The paper focuses on the problem arising when the follower parameters $D^F$ contains sensitive information $D^S$ that should not be revealed to the leader. 
In the case of electricity and natural gas markets, these parameters may represent the costs of producers or the demand profile of consumers. If released, they can provide a competitive advantage to the leader, or other strategic agents in the energy system, and may result in financial losses for the follower.

The \emph{privacy goal} is to ensure that the sensitive information $D^S$ contained in the follower parameters $D^F$ is not breached during the coordination process described in Problem \eqref{eq:hierarchical}. The next section introduces a formal notion that will be used to achieve this goal.

\subsection{Differential Privacy}
\label{sec:differential_privacy}

\emph{Differential privacy} \cite{dwork:06} (DP) is a rigorous privacy notion used to protect the participation disclosure of an individual in a computation. 
A randomized mechanism $\cM \!:\! \D \!\to\! \R$ with
domain $\D$ and range $\R$ is $\epsilon$-differential private if, for
any output response $O \subseteq \R$ and any two \emph{neighboring}
datasets $D, D' \in \D$ differing in at most one individual (written
$D \sim D'$),
\begin{equation}
  \label{eq:dp_def} 
  Pr[\cM(D) \in O] \leq \exp(\epsilon)\, Pr[\cM(D') \in O].
\end{equation}
\noindent 
The parameter $\epsilon \geq 0$ is the \emph{privacy loss} of the
mechanisms, with values close to $0$ denoting strong privacy.  

DP satisfies several important properties, including \emph{composability} and immunity to \emph{post-processing}. Composability ensures that a combination of differentially private mechanisms preserve privacy. 
\begin{theorem}[Sequential Composition]
\label{th:seq_composition}
The composition $(\cM_1(\bm{D}), \ldots, \cM_k(\bm{D}))$ of a collection $\{\cM_i\}_{i=1}^k$ of $\epsilon_i$-differential private mechanisms satisfies $(\sum_{i=1}^{k} \epsilon_i)$-differential privacy.
\end{theorem}

Immunity to post-processing, ensures that
privacy guarantees are preserved by arbitrary post-processing steps.
\begin{theorem}[Post-Processing \cite{dwork:13}] 
	\label{th:postprocessing} 
	Let $\cM$ be an $\epsilon$-differential private mechanism and $g$ be an arbitrary mapping from the set of possible outputs to an arbitrary set. Then, $g \circ \cM$ \mbox{is $\epsilon$-differential private.}
\end{theorem}

In private data analysis settings, 
a function $Q$ (also called \emph{query}) from a data set $D \in \D$ to a result set $R \subseteq \RR^n$ 
can be made differentially private by injecting random noise to its
output. The amount of noise depends on the \emph{sensitivity} of the
query, denoted by $\Delta_Q$ and defined as
\(
\Delta_Q = \max_{D \sim D'} \left\| Q(D) - Q(D')\right\|_1.
\)
In other words, the sensitivity of a query is the maximum
$l_1$-distance between the query outputs of any two neighboring
datasets $D$ and $D'$.

While the classical DP notion protects the individuals presence 
into a dataset, many applications involve components whose {\em
presence} is public information. In the problem of interest, for instance, the gas market participants are known and have non-negative demands. However, the demand values are highly sensitive, as mentioned in the previous section. 
To protect the values associated with these components, the concept of \emph{indistinguishability} is reviewed.

Indistinguishability was introduced in \citep{andres2013geo} to protect user locations in the Euclidean plane and then generalized
in \cite{chatzikokolakis2013broadening,koufogiannis:15} to arbitrary
metric spaces.  Consider a dataset $D$ to which $n$ individuals
contribute their data $x_i$ and $\alpha > 0$. An adjacency relation
that captures the data variation of a single individual is defined as:
\begin{align*}
  D \sim_\alpha D' \iff \exists i: d(x_i, x_i') \leq \alpha \land\; \forall j \neq i: d(x_j, x_j') = 0,
\end{align*}
where $d$ is a distance function on $\D$.  Such adjacency definition
is useful to hide individual participation up to some quantity
$\alpha$. In our application, $\alpha$-indistinguishability allows customers to reveal gas demand profiles that hide the real consumption by a factor of $\alpha$.

Given $\epsilon \geq 0, \alpha > 0$, a randomized mechanism $\cM : \D \to \R$ with domain $\D$ and range $\R$ is \emph{$\alpha$-indistinguishable $\epsilon$-differentially private} ($\alpha$-indistinguishable for short) if, for any event $O \subseteq \R$ and any pair $D \sim_\alpha D'$, ($D,D' \in \D$), Equation \eqref{eq:dp_def} holds.

\section{The PPS Problem}
\label{sec:PPS_problem}

The Privacy-Preserving Stackelber (PPS) problem establishes the fundamental desiderata to be delivered by an obfuscation mechanism. It operates on the follower sensitive parameters $D^S$ exchanged to ensure coordination between the leader and the follower in resolution of Problem $S$ (\Cref{eq:hierarchical}).

\begin{definition}
\label{def:PPS}
Given a Stackelber game $S$ and positive real numbers $\alpha, \eta$, the PPS problem produces a problem $\tilde{S}$ such that:
\begin{enumerate}
	\item \emph{Privacy}: \label{cond:1}
	$\tilde{S}$ satisfies $\alpha$-indistinguishability for the sensitive parameters $D^S$.
	\item \emph{Fidelity}: \label{cond:2}
    The privacy-preserving optimal objective value of each agent, i.e. the leader and/or the follower, is \textit{close} to its original optimal objective value, and the variation in objective value is bounded by parameter $\eta$.
    \item \emph{Consistency}: \label{cond:3}
    The privacy-preserving follower subproblem satisfies the problem constraints \eqref{eq:hierarchical3}.
\end{enumerate}


\end{definition}
The fidelity and consistency conditions make sure that the privacy-preserving coordination mechanism is feasible, and preserve the original cost of the agents.

Prior to introducing a mechanism that achieves the desired conditions above, the paper discusses the application of interest in details: The coordination of sequential electricity and natural gas markets. 

\section{Stackelberg game for the coordination of electricity and gas markets} \label{sec:EG_markets}


The coordination between electricity and natural gas markets can be modeled as a Stackelberg game between a leader, i.e. the \emph{electricity unit commitment} (UC), and a follower, i.e. the sequential clearing of the \emph{electricity market} (EM) and the \emph{natural gas market} (GM), as recently proposed in \cite{byeon2019unit}. This Stackelberg coordination problem is schematically depicted at the bottom of Figure \ref{fig:elec_gas_markets} (which excludes the privacy-preserving mechanism), and referred to as \emph{gas-aware UC} (GAUC). It  represents the sequential order of the agents decision making and the interdependencies between them. 
\begin{figure*}[t]
    \centering
    \includegraphics[width=0.85\linewidth]{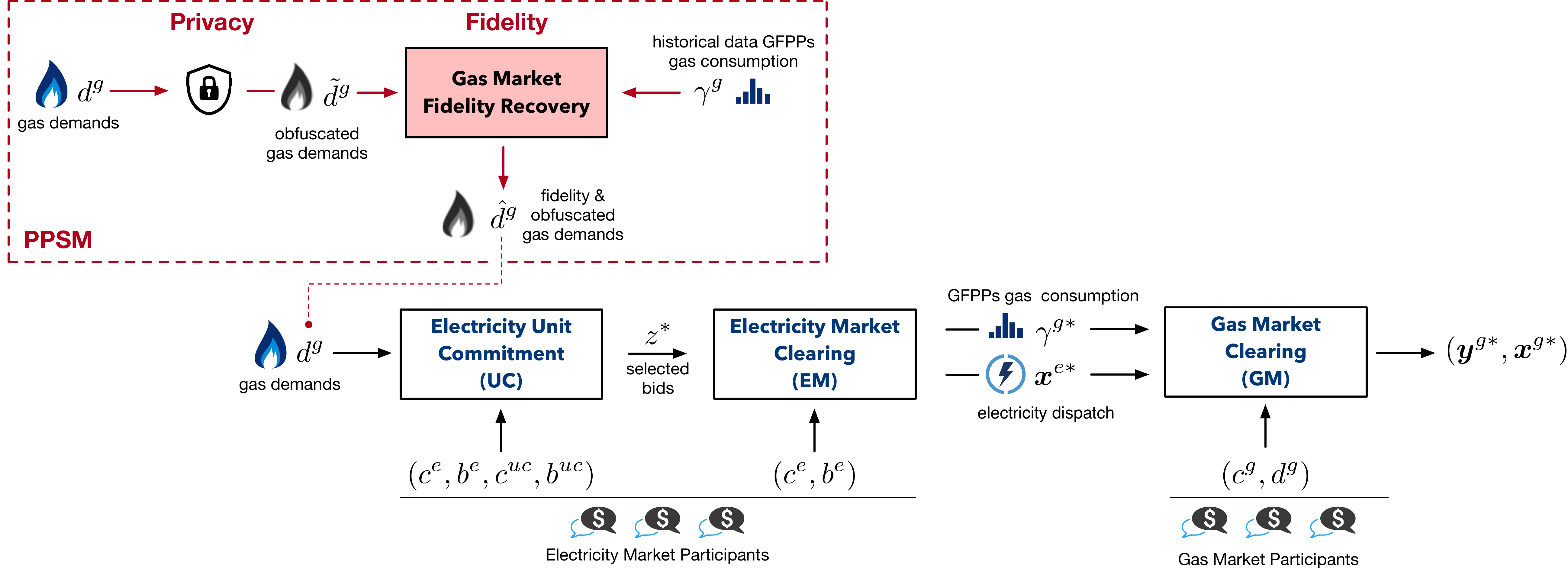}
    \caption{Sequential clearing of electricity unit commitment (UC), electricity and natural gas markets.}
    \label{fig:elec_gas_markets}
\end{figure*}

\subsubsection*{\textbf{GAUC Coordination variables}~($\bm{z}$)}
The GAUC is responsible for committing the bids of the electricity market participants, while anticipating the impact of these decisions on the follower, i.e. the sequential EM and GM clearings. Due to the participation of gas-fired power plants (GFPP)s at the interface between electricity and natural gas systems, the \emph{commitment} $\bm{z}^*$ of the selected bids of GFPPs by the UC impacts both the EM and GM clearings. 

Firstly, the EM uses the bids commitment $\bm{z}^*$ of the EM participants as input and dispatches these \textit{selected bids} to maximize the social welfare of the power system. Secondly, the GM dispatches its participants bids in order to maximize the social welfare in the natural gas system.

\subsubsection*{\textbf{EM and GM Coordination variables}~($\bm{x}^e$ and $\bm{y}^g$)}
The electricity dispatch $\bm{x}^{e^*}\!\!$ of GFPPs is directly linked to their gas consumption $\gamma^{g^*}\!\!$, which is used as input to the GM clearing problem. In return, the reaction of the follower, i.e., the electricity dispatch $\bm{x}^{e}$ and gas prices $\bm{y}^{g}$, are accounted in the leader decisions through the \textit{objective cost} $\cO^{uc}$ and feasible space of the UC problem $\cZ^{uc}$. 

Firstly, the UC problem objective accounts for the dispatch cost of the electricity bids. 
Secondly, the UC problem feasible space accounts for \emph{bid-validity constraints} which embed the interdependencies 
between the gas prices $\bm{y}^g$ and the marginal 
electricity production cost of GFPPs.
In practice, these constraints enforce the price of the last selected GFPPs bid to be no larger than their marginal electricity production cost. These bidirectional interdependencies between the leader and the follower represent a classic Stackelberg game.

\subsubsection*{\textbf{GAUC problem}}
Therefore, the compact formulation of the GAUC as a hierarchical optimization problem is as follows:
\begin{subequations} \label{3level}
	\begin{alignat}{7}
	\cP^{uc} = 
	& \min_{\overset{\bm{z}\in\{0,1\}^N}{\underset{\bm{x}^e ,  \bm{y}^g \geq \bm{0}}{}}} 
	&&  \;
	\cO^{\text{uc}} = 
	c^{uc^\top}\!\bm{z}  + c^{e^\top}\!\bm{x}^e \label{3level0} \\
	& \quad \text{s.t.} && \bm{z} \in \cZ^{uc} \label{3level1} \\
	& \quad                && A^{uc}  \bm{z}    +     B^{uc}  \bm{y}^g    \geq   \ b^{uc} \label{3level2} \\
	& \quad                &&  \bm{x}^e , \bm{y}^g              \in      \text{primal and dual sol. of } \eqref{bilevel_heat_elec},  \label{3level3}
	\end{alignat}
\end{subequations}
where $N$ is the dimensionality of the commitment vector. 
The leader's objective aims at minimizing the electricity system operating cost, which includes \emph{no-load} and \emph{start-up} costs $c^{uc^\top}\!\bm{z}$, and the \emph{cost of dispatching} the electricity bids $c^{e^\top}\!\bm{x}^e$, constrained by the techno-economic characteristics of electricity suppliers \eqref{3level1}, and the bid-validity constraints \eqref{3level2}. 

\subsubsection*{\textbf{EM and GM Clearing Problems}}
Furthermore, due to the sequential order of EM and GM clearings, the follower's problem \eqref{3level3} can be expressed as a hierarchical optimization problem, such that:
\begin{subequations} \label{bilevel_heat_elec}
\begin{alignat}{7}
\cP^{e} = 
	& \min_{\bm{x}^e ,  \bm{y}^g}  \;
		&& \cO^{e} = c^{e^\top}\!\bm{x}^e  
		\label{bilevel_heat_elec0} \\
		& \quad    \text{s.t.}   && \bm{x}^e \in \cZ^e  
		\label{bilevel_heat_elec1}  \\
		& \quad     &&  A^{e}  \bm{x}^e + B^{e}    \bm{z} \geq b^{e}
		\label{bilevel_heat_elec1.2} \\
		& \quad  	&&   \bm{y}^g   \in   \text{dual sol. of }
	&&  \cP^{g} =  \min_{\bm{x}^g} \;
		\cO^{g} = c^{g^\top}\!\bm{x}^g 
		\tag{5a} \label{bilevel_heat_elec2}  \\
		& \quad 	&&  \quad   &&  \  
		\text{s.t. } \quad \bm{x}^g  \in \cZ^g \tag{5b}
		\label{bilevel_heat_elec4} \\
		& \quad 	&&  \quad   &&  \  
		\phantom{s.t. } \quad  A^{g}  \bm{x}^g   +  B^{g}  
		\bm{x}^e = d^{g}. \tag{5c}\label{bilevel_heat_elec5} 
\end{alignat}
\end{subequations}
\setcounter{equation}{5}
The objective \eqref{bilevel_heat_elec0} of the middle-level problem is to minimize the electricity dispatch cost, constrained by linearized power flow constraints \eqref{bilevel_heat_elec1} and bounds on the selected bids $\bm{z}$ of electricity suppliers \eqref{bilevel_heat_elec1.2}.
The lower-level problem \eqref{bilevel_heat_elec2}--\eqref{bilevel_heat_elec5} represents the GM clearing, which seeks to minimize the natural gas dispatch cost \eqref{bilevel_heat_elec2}, constrained by linearized gas flow constraints and the bounds on natural gas supply and demand bids \eqref{bilevel_heat_elec4}, as well as the nodal gas balance equation \eqref{bilevel_heat_elec5}. 

The detailed expressions of these optimization problems, as well as the matrices $A^{uc}$, $B^{uc}$, $A^{e}$, $B^{e}$, $A^{g}$, $B^{g}$, the vectors $c^{uc}$, $b^{uc}$, $c^{e}$, $b^{e}$, $c^{g}$, $d^{g}$, and the polytopes $\cZ^{uc}$, $\cZ^{e}$, $\cZ^{g}$ is derived from \cite{byeon2019unit} and provided in the online Appendix \cite{onlineappendix}.

\subsubsection*{\textbf{Public and sensitive information}}
As illustrated in Figure \ref{fig:elec_gas_markets} (bottom), the UC problem takes as input the techno-economic characteristics of electricity 
suppliers, which are represented in Problem \eqref{3level} by the 
matrices $A^{uc}$ and $B^{uc}$, the vectors $c^{uc}$ and $b^{uc}$, 
and the polytope $\cZ^{uc}$.
As the UC and EM cover the same energy system, these two agents also 
take the same parameters as input, which include the price $c^e$ and quantity $b^e$ bids of electricity suppliers, electricity demand profile, and electricity network technical characteristics. 
These parameters are represented in Problem \eqref{bilevel_heat_elec} by the matrices $A^{e}$ and $B^e$, the vectors $c^e$ and $b^e$, and the polytope $\cZ^{e}$. Moreover, the UC takes as input the parameters of the the GM 
follower, which include the price $c^g$ and quantity bids of gas suppliers, gas demand profile $d^g$, and gas network technical characteristics. These parameters are represented in Problem \eqref{bilevel_heat_elec} by the matrices $A^g$ and $B^g$, the vectors $c^g$ and $d^g$, and the polytope $\cZ^{g}$. 
In particular, the gas demand profile, $d^g = \left[d_{j}^g:\forall j \in \mathcal{V} \right]$ at all nodes $j \in \mathcal{V}$ of the gas network, represents the sensitive information (referred to as $D^S$ in Problem (1)) of the GM follower.


In contrast, the following outcomes of the GM clearing are traditionally considered \emph{publicly available} information: the original objective value $\cO^{g^*}$ and natural gas prices $\by^{g^*}$. The mechanism introduced next will leverage this public information to restore fidelity of the PPS problem.

\section{The privacy-preserving Stackelberg mechanism (PPSM)}
\label{sec:PPSM}

The \emph{Privacy-Preserving Stackelberg Mechanism} (PPSM) aims at protecting the \emph{privacy} of the gas demand profile $d^g$ and achieve the \textit{consistency} and \textit{fidelity} of the obfuscated data $\hat{d}^{g}$ in the context of the coordination between electricity and natural gas markets.

\subsection{Overview} 

A schematic representation of the mechanism is provided in
at the top of Figure \ref{fig:elec_gas_markets}. PPSM operates in three phases, described as follows:
\begin{enumerate}
    \item The \emph{privacy phase} takes as input the original, sensitive gas profile $d^g$, and produces a new, \emph{obfuscated}, natural gas demand profile $\tilde{d}^g$ that is $\alpha$-indistinguishable from $d^g$.
    
    \item The \emph{fidelity phases} redistribute the noise introduced in the privacy phase to obtain a privacy-preserving profile $\tilde{d}^g$ and produces a new, privacy-preserving, gas demand profile $\hat{d}^g$ that satisfies the original GM constraints and renders the GM objective $\tilde{\cO}^{{g}^*}$ and/or the GAUC objective $\tilde{\cO}^{{uc}^*}$ \emph{faithful} to their original counterparts, ${\cO}^{{g}^*}$ and $\cO^{{uc}^*}$.

    \item Finally, the mechanism uses the privacy-preserving and faithful gas demand profile $\hat{d}^g$ as input to solve the GAUC problem introduced in the Section \ref{sec:EG_markets}.
\end{enumerate}

\subsection{Privacy Phase}

The PPSM privacy phase takes as input the original vector of gas demand profile $d^g$ and constructs a privacy-preserving version $\tilde{d}^g$ using the Laplace mechanism. 
  
\begin{theorem}[Laplace Mechanism \cite{dwork:06}]
\label{th:m_lap} 
Let $Q$ be a numeric query that maps datasets to $\RR^n$. The Laplace mechanism that outputs $Q(D) + \xi$, where $\xi \in \RR^n$ is drawn from the Laplace distribution $\Lap(\Delta_Q / \epsilon)^n$, achieves $\epsilon$-differential privacy.
\end{theorem}

\noindent
In the above, $\Lap(\lambda)$ denotes the Laplace distribution with 0
mean and scale $\lambda$, and $\Lap(\lambda)^n$ denotes the
i.i.d.~Laplace distribution over $d$ dimensions with parameter
$\lambda$. The Laplace mechanism with parameter
$\lambda=\alpha/\epsilon$ satisfies
$\alpha$-indistinguishability \cite{chatzikokolakis2013broadening}. 
As a result, the privacy-preserving gas profile values $\tilde{d}^g$ are obtained as follows:
\begin{equation} \label{eq:setp1}
	\tilde{d}^g  = d^g + \Lap(\alpha/\epsilon)^N,
\end{equation}
where $N$ is the appropriate dimensionality of the gas profile vector, $\alpha >0$ is the \emph{indistinguishability} value (e.g., the amount of MWh we want to protect), and $\epsilon \geq 0$ is the privacy loss (with typical values within $\ln(1.1)$ to $\ln(10)$). 
Importantly, the Laplace mechanism has been shown to be \emph{optimal}: it minimizes the mean-squared error for identity queries with respect to the L1-norm \cite{koufogiannis:15}.


While \eqref{eq:setp1} ensures $\alpha$-indistinguishability in the 
$\epsilon$-DP framework, the obfuscated data may not achieve \textit{consistency} or strong \emph{fidelity} with respect to the original problem. Crucially, demand profiles generated by this mechanism often fail to produce a feasible solution to the GAUC problem, as illustrated in Section \ref{sec:experiments}.

\subsection{Fidelity Phase} 
\label{sec:fidelity_phase}

To remedy such limitation, the proposed PPSM introduces an \emph{optimization-based post-processing step} that aims at establishing consistency and fidelity with respect to the agents constraints and objectives. 
The goal of the fidelity phase is that of producing a new gas demand profile $\hat{d}^g$ that satisfies the GM constraints and is faithful to the GM and/or GAUC objective values. Indeed, while both agents have competing objectives, the PPSM should provide fidelity guarantees to both agents in order to be applicable.

Therefore this paper introduces two fidelity approaches, one in which the fidelity criteria focuses on the GM cost, and another where it focuses on the GAUC cost. As previously discussed, the original GM cost $\cO^{g^*}$ is traditionally public information, therefore the PPSM can directly restore fidelity on this value.
However, due to the sequential order of the GAUC and GM clearing, the original GAUC cost $\cO^{uc^*}$ is only revealed after the privacy-preserving data $\hat{d}^g$ has been released. Therefore, fidelity on this value cannot directly be restored.
Instead, the original coordination variables, i.e. the natural gas prices $\by^{g^*}$ are public information. As these coordination variables provide feedback from the GM problem to the GAUC problem, it is expected that restoring fidelity on these variables ensures fidelity on the GAUC objective. 
These two fidelity approaches are detailed below.

\subsubsection{\textbf{GM cost: Primal fidelity ($PPSM_p$)}}
\label{sec:p_fidelity_phase}

The first post-processing phase 
aims at restoring feasibility of the natural gas problem and establishing fidelity of the \emph{privacy-preserving} GM objective cost $\hat{\cO}^g = c^{g^\top}\!\hat{\bx}^{g}$ with respect to its 
counterpart $\cO^{g^*}$ computed using the original demand profile. 

However, the solutions and objective of GM problem are affected not only by the gas demand $d^g$, but also by the gas consumption $\bx^e$ of GFPPs, which is the solution of the EM problem (see \Cref{bilevel_heat_elec5}). As highlighted in Figure \ref{fig:elec_gas_markets}, as the GAUC and EM are cleared after the release of the privacy-preserving data $\hat{d}^g$, historical data $\bar{x}^e$ can be used as input to the optimization-based post-processing phase in place of a solution $\bx^{e^*}$ to the EM problem. In practice, past observations on the gas and electricity demand profile 
can be publicly used to accurately predict the future gas consumption of GFPPs, denoted $\hat{x}^{e}$. For the remainder of this paper, it is assumed that these historical observations provide an accurate prediction of the future gas consumption.

The Primal fidelity post-processing phase redistributes the noise introduced by the Laplace mechanism on the obfuscated demand profile $\tilde{d}^g$ to generate a new load profile $\hat{d}^g$ through the following bilevel optimization program:
\begin{subequations} \label{eq:pp2}
\begin{align}
	\min_{\hat{\bd}^g, \hat{\bx}^g} \;\;&
		\| \hat{\bd}^g - \tilde{d}^g\|^2_2 
		\label{eq:pp2_0} \\
	 \text{s.t.} \;\;&
	| c^{g^\top}\!\hat{\bx}^g - \mathcal{O}^{{g}^*} | \leq \eta 
	\label{eq:pp2_1} \\
	& \hat{\bx}^g = \text{ primal sol. of } 
	\cP^g(\bar{x}^{e},\hat{\bd}^g). \label{eq:pp2_2}
\end{align}
\end{subequations}

The problem minimizes the distance of the post-processed demand profile $\hat{d}^g$ to the Laplace-obfuscated ones $\tilde{d}^g$ \eqref{eq:pp2_0}, while enforcing fidelity of the objective cost with respect to the original loads profiles \eqref{eq:pp2_1}. 
The public data $\mathcal{O}^{g^*}$ represents the objective cost of the GM clearing problem $\mathcal{P}^g\left(\hat{x}^{e},d^g\right)$ with the original demand profile, and collected from historical data. 
The gas dispatch variables $\hat{\bm{x}}^g$ are defined in the lower-level problem \eqref{eq:pp2_2} as the solution of the GM clearing problem $\mathcal{P}^g(\hat{x}^{e}, \hat{\bd}^g)$ (see \Cref{bilevel_heat_elec2,bilevel_heat_elec4,bilevel_heat_elec5}) with the post-processed demand profile $\hat{\bd}^g$. 
Using the equivalent Karush-Kuhn-Tucker (KKT) conditions of the linear lower-level problem \eqref{eq:pp2_2} and the Fortuny-Amat linearization, this bilevel problem can be recast as a mixed integer linear program (MILP) \cite{gabriel2012complementarity}. 
The reformulation is detailed in the online Appendix \cite{onlineappendix}.

\subsubsection{\textbf{GAUC cost: Dual fidelity ($PPSM_d$)}}
\label{sec:d_fidelity_phase}

While the fidelity phase described above enforces fidelity of the GM objective with respect to the original demand profile, it ignores the impact of the post-processed data on the UC objective (leader) problem. 

To address this limitation, this section introduces an additional fidelity algorithm that enforces fidelity of the coordination variables $\bm{y}^g$. 
Recall that the GM clearing problem provides feedback to the UC problem via these coordination variables. Therefore, this approach allows us to enforce fidelity on both the UC and GM costs. The resulting fidelity phase redistributes the noise introduced by the Laplace mechanism on the obfuscated demand profile $\tilde{d}^g$ to generate a new load profile $\hat{d}^g$ through the following bilevel optimization program:
\begin{subequations} \label{eq:pp3}
\begin{align}
	\min_{\hat{\bd}^g, \hat{\bx}^g, \hat{\by}^g} \;\;& 
	\| \hat{\bd}^g - \tilde{d}^g\|^2_2 \label{eq:pp3_0} \\
	\text{s.t.}\;\;& 
	|  \hat{\by}^g - \bar{y}^{g^*} | \leq \eta \label{eq:pp3_1b} \\
& \{\hat{\bx}^g , \hat{\by}^g \} = \text{ primal and dual sol. of } \mathcal{P}^g(\bar{x}^e, \hat{\bd}^g). \label{eq:pp3_2}
\end{align}
\end{subequations}

Similarly to the \emph{primal fidelity}, this problem seeks to minimize the distance between the post-processed demand $\hat{d}^g$ to the Laplacian obfuscated one $\tilde{d}^g$ \eqref{eq:pp3_0}. 
At the same time, it enforces fidelity of the coordination variables with respect to the original loads profiles  \eqref{eq:pp3_1b}. 
The \emph{public data} $\bar{y}^{g^*}$ represents the natural gas prices, associated with the original demand profile. In practice, in order to ensure transparency of energy markets, this information is indeed publicly revealed by market operators, as it represents the price at which energy is bought and sold at each node of the network.
The natural gas dispatch variables $\bx^g$ and prices $\hat{\by}^g$ are defined in the lower-level problem \eqref{eq:pp3_2} as the primal and dual solutions of the GM clearing problem $\mathcal{P}^g(\bar{x}^e, \hat{d}^g)$ with the post-processed demand profile. Similarly to the primal fidelity problem \eqref{eq:pp2}, this bilevel problem can be recast as an MILP.

\begin{theorem}
For a given $\epsilon \geq 0$,  $\alpha >0$, and $\eta >0$, 
PPSM satisfies the conditions of the PPS problem (\Cref{def:PPS}).
\end{theorem}

\begin{proof}
The privacy phase of PPSM produces a gas demand profile $\tilde{d}^g$ using the Laplace mechanism with parameter $\lambda = \alpha / \epsilon$ (\Cref{eq:setp1}). By \Cref{th:m_lap}, the resulting demand profile satisfies $\alpha$-indistinguishability for the given privacy parameter $\epsilon$. Therefore, PPSM satisfies the \emph{privacy} requirement of the PPS problem (\Cref{def:PPS}).

Notice that the PPSM fidelity phase takes as input the privacy-preserving gas profiles $\tilde{d}^g$ generated during the privacy phase, and uses exclusively the public information---e.g., the problem model and the historical data $\hat{x}^e$ (primal fidelity), or the natural gas prices $\bar{y}^{g^*}$ (dual fidelity). 
Therefore, by the post-processing immunity property of differential privacy (\Cref{th:postprocessing}) the output $\hat{d}^g$ of the fidelity phase satisfies $\alpha$-indistinguishability.

The PPS \emph{fidelity} and \emph{consistency} requirements are ensured directly by the PPSM fidelity step. In particular the fidelity requirement is satisfied by Constraints \eqref{eq:pp2_2} and \eqref{eq:pp3_1b} in the primal and dual fidelity steps, respectively. The consistency requirement is satisfied by Constraints \eqref{eq:pp2_2} and \eqref{eq:pp3_2}, in the primal and dual fidelity steps, respectively. 
\end{proof}

\section{The Cost of Privacy}

The section analyzes the \emph{cost of privacy} as the impact of the data perturbation induced by the privacy-preserving mechanism and the optimal objective cost of the follower agent in the privacy-preserving problem. 
The results below hold under the (very mild) assumption that the historical values $\bm{\hat{x}^{e}}$ used in the primal and dual fidelity phases are accurate. 

\begin{theorem}
\label{thm:error}
  After the fidelity phase, the expected error induced by $PPSM_p$ or $PPSM_d$ on the gas demand values $\hat{d}^g$ is bounded by the inequality: 
  $$
    \EE\left[ \| \hat{\bd}^{g^*} - d^g \| \right] \leq 
    4 \dfrac{\alpha^2}{\epsilon^2}, 
  $$
  where $\hat{\bd}^{g^*}$ is the solution to Problem \eqref{eq:pp2} or \eqref{eq:pp3}, and
  $d^g$ is the original profile. 
\end{theorem}

The proof of the theorem above is analogous, in spirit, to the proof of Theorem 5 in \cite{fioretto:CPAIOR-18}.

\begin{proof}
  Denote with $\tilde{d}^g = d^g + \Lap(\alpha/\epsilon)$ the privacy-preserving version of the gas demand profile obtained by the Laplace mechanism. We have that:
  \begin{align*}
  \| \hat{\bd}^{g^*} - d^g \|_{2} &\leq \| \hat{\bd}^{g^*} - \tilde{d}^g \|_{2} + \| \tilde{d}^g - d^g \|_{2} \label{eq:p2}\\
            &\leq 2 \|\tilde{d}^g - d^g\|_{2} \leq 4 \dfrac{\alpha^2}{\epsilon^2}.
  \end{align*}
  The first inequality follows from the triangle inequality on norms. 
  The second inequality follows from:
  $$
      \| \hat{\bd}^{g^*} - \tilde{d}^g \|_{2} \leq \|\tilde{d}^g - d^g\|_{2}
  $$
  by optimality of 
  $\langle \hat{\bd}^{g^*}, \hat{\bx}^*\rangle$ 
  and the fact that 
  $\langle d^g, \bx^*\rangle$ is a feasible solution to constraints 
  \eqref{eq:pp2_1} and \eqref{eq:pp2_2} (Problem \ref{eq:pp2}), 
  or, similarly, by optimality of
  $\langle \hat{\bd}^{g^*}, \hat{\bx}^*, \hat{\by}^*\rangle$ 
  and the fact that 
  $\langle d^g, \bx^*, \by^*\rangle$ is a feasible solution to constraints 
  \eqref{eq:pp3_1b} and \eqref{eq:pp3_2} (Problem \ref{eq:pp3}).
  The third inequality follows directly from the variance of the Laplace distribution.
\end{proof}

\begin{theorem} \label{th2}
	After the fidelity phase, the error induced by $PPSM_p$ on the follower's objective is bounded in expectation by:
  $$ 
  \EE \left[ | \cO^{g*} - \hat{\cO}^{g^*} | \right] \leq 
  \dfrac{2\sqrt{2}\alpha}{\epsilon}|\bm{y}^{g^*} + 
  \eta|,
  $$  
  where $\bm{y}^{g^*} $ represents the natural gas prices of the GM clearing problem with the original demand profile $d^g$.
\end{theorem}
While the $PPSM_p$ explicitly bounds the cost of privacy in \eqref{eq:pp2_1}, Theorem \ref{th2} shows that the $PPSM_p$ also indirectly provides bounds on the cost of privacy, which are proportional to the sensitivity of the objective value with respect to the demand.

\begin{proof}
The sensitivity of the objective function of the follower with 
respect to a small perturbation on the demand is given by the 
dual variable $\bm{y}^g$ associated with constraint \eqref{bilevel_heat_elec5}. As a result, for a small difference between the original and privacy-preserving gas demand $| d^{g} - \hat{d}^{g}|$, the 
following can be derived:
\begin{equation*}
  \frac{ \cO^{g*} - \hat{\cO}^{g^*} }{ d^{g} - \hat{d}^{g} }= \bm{\hat{y}}^{g^*}.
\end{equation*}
Additionally, from \Cref{thm:error} the variation on demand profile $| d^{g} - \hat{d}^{g} |$ in the PPSM is bounded in expectation by $\dfrac{2\sqrt{2}\alpha}{\epsilon}$. Therefore, the difference in optimal objective costs is bounded by
\begin{equation}
\mathbb{E} \left[| \cO^{g*} - \hat{\cO}^{g^*} |\right]
			    \leq 
			    \dfrac{2\sqrt{2}\alpha}{\epsilon}
			    |\bm{y}^{g^*} + \eta|.
\end{equation}
The right-hand side of the expression above follows from the fidelity constraint \eqref{eq:pp3_1b}.
\end{proof}

\section{Experimental Evaluation}
\label{sec:experiments}

\def\Mp{\textsl{PPSM}_{primal}}
\def\Md{\textsl{PPSM}_{dual}}

\subsection{Case study setup}

The PPSM is evaluated on a test system which is representative of the joint natural gas and electricity systems in the Northeastern US \cite{bent2018joint,byeon2019unit}. 
The system is composed of the IEEE 36-bus NPCC electric power
system \cite{allen2008combined} and a gas transmission network
covering the Pennsylvania-To-Northeast New England area.
The unit commitment data for the $91$ power plants is derived from the RTO unit commitment test system \cite{krall2012rto} based on their technology.

This case study analyzes the performance of the PPSM under various operating conditions of the gas and electricity systems. The electricity demand profile is uniformly increased by $30\%$ and $60\%$, and the gas demand profile is uniformly increased by $10\%$ up to
$130\%$, producing increasingly stressed and difficult operating conditions. 
Both the primal fidelity and dual fidelity approaches are compared to the standard Laplace mechanism (which was shown to be optimal for identity queries \cite{koufogiannis:15}, like those used in this work), for varying values of the indistinguishability parameter $\alpha \in \{0.1, 1, 10\}\times 10^2$ MWh, and the fidelity parameter $\eta \in \{0.01, 0.1, 10.0\}\%$ of the original objective value $\cO^{g^*}$. Notice that the maximal original gas demand vector $d^g$ in the highest stress factor adopted, has minimum, median, and maximum values: $0$, $3.38 \times 10^2$, $98.31 \times 10^2$, respectively. Therefore, the indistinguishability parameters adopted ensure a very low privacy risk. 

We generate $30$ repetitions for each test case and report average results in all experiments. The privacy-preserving mechanisms (PPSM phases 1 and 2) are implemented in Python 3.0 with Gurobi 8.1. The GAUC problem (phase 3) uses the c++ implementation of \cite{byeon2019unit}. A wall-clock limit of 1 hour is given to generate and solve each of the instances. The resolution of the privacy-preserving demand profiles (phases 1 and 2 of PPSM) takes less than 30s for any of the instances. 


\begin{table}[!t]
\centering
\resizebox{\linewidth}{!}
{
\begin{tabular}{llrr @{\hspace{15pt}}|@{\hspace{15pt}} rrr}
\toprule
$\cM$ & $\alpha$ & sat.(\%) & $\Delta_d$ (L1) & $\Delta_{\cO^{uc}} (\%)$ &   $\Delta_{\cO^e} (\%)$&   $\Delta_{\cO^g} (\%)$ \\
\midrule
\multirow{3}{*}{Laplace}
       & 0.1   &  72.61 &   5.96  & 0.1564 &   0.3666 &   0.3668 \\
       & 1.0   &  19.76 &  59.62    &  1.6065 &   3.4445 &   3.4446 \\
       & 10.0  &   4.28 & 596.28    & 21.0501 &  43.1685 &  43.1685 \\
\cline{2-3}
\cline{4-7}
\multirow{3}{*}{$PPSM_p$}
       & 0.1   & 99.29 &  3.89 &    0.0139 &   0.0564 &   0.0565 \\
       & 1.0   & 89.90 & 13.26 &    0.0570 &   0.1067 &   0.1068 \\
       & 10.0  & 82.38 & 14.32 &    0.1235 &   0.2308 &   0.2308 \\
\cline{2-3}
\cline{4-7}
\multirow{3}{*}{$PPSM_d$}
       & 0.1   & 99.28 &  3.89 &   0.0130 &   0.0705 &   0.0706 \\
       & 1.0   & 96.38 & 13.26 &   0.0536 &   0.1064 &   0.1065 \\
       & 10.0  & 93.95 & 14.32 &   0.1193 &   0.2256 &   0.2256 \\
\bottomrule
\end{tabular}
}
\caption{Left: Satisfactory instances (\%) and L1 errors (MWh) on the gas demands for varying indistinguishability parameters $\alpha$, and $\eta = 0.1 \%$ of the GM cost. 
Right: Errors (\%) on the leader objective ($\cO^{uc}$) and followers' objective ($\cO^e$ and $\cO^g$).
\label{tab:1}}
\end{table}

\subsection{Limits of The Laplace Mechanisms}
\label{sec:laplace_limits}

This section studies the applicability to our context of interest of the Laplace mechanism. 
Table \ref{tab:1} (left) reports the percentage of the feasible solutions (over $1170$ instances) across different values of the indistinguishability parameter $\alpha$. It compares the Laplace mechanism with $PPSM_p$ and $PPSM_d$ that use, respectively, the primal and the dual fidelity phases defined in \Cref{sec:p_fidelity_phase} and \Cref{sec:d_fidelity_phase}. When the indistinguishability parameter exceeds $0.1$ the Laplace-obfuscated gas demands rarely produce a feasible solution to the GAUC problem. 
\emph{These results justify the need of studying more advanced privacy-preserving mechanisms for Stackelberg game problems, and hence the proposed PPSM}. 
In contrast, the PPSMs result in a much higher number of feasible solutions. Indeed, all ``unsolved'' PPSM cases are due to timeout and not as a direct result of infeasibility.  
Additionally, we verified that the two PPSMs are always able to find a feasible solution to the GM problem $\cP^g(\bar{x}^e, \hat{\bd}^g)$. 

Table \ref{tab:1} (left) also reports the L1 distance $\Delta_d$ between the original gas demand $d^g$ and the privacy-preserving versions obtained with each of the mechanisms analyzed. Unsurprisingly, the L1 errors increase with the increasing of the indistinguishability parameter $\alpha$, as larger indistinguishability induce more noise.
However, the L1 errors introduced by the PPSM are much more contained than the Laplace ones, producing more than a order of magnitude more accurate results for the larger indistinguishability parameters.
\emph{These results indicate that the highly-perturbed demand profiles induced by the Laplace mechanism lead to infeasibility in the GAUC and GM problems, whereas the PPSMs manage to restore \textit{consistency} of the post-processed demand profiles.}

\subsection{Leader and Follower Objectives}

The next results evaluate the ability of PPSM to preserve the optimal objective values of the leader and the follower agents problems.  

Table \ref{tab:1} (right) tabulates the errors, in percentage, on the objective costs of the GAUC problem (leader), the EM problem and the GM problem (followers) at varying indistinguishability parameters $\alpha$, and for a fixed fidelity parameter $\eta = 0.1\%$.
The errors $\Delta_{\cO}$ are defined as $\frac{|\cO^* - \tilde{\cO^*}|}{\cO} \%$, where $\cO \in \{\cO^{uc}, \cO^{e}, \cO^{g}\}$ and are computed in expectation over the feasible instances only. 
Parameter $\alpha$ controls the amount of noise being added to the gas demand profiles, therefore, the objective costs are closer to their original values when $\alpha$ is small. 
\emph{Observe that the PPSMs induce objective costs differences that are from one to two orders of magnitude more accurate that those induced by the Laplace mechanism, and that are at most $1\%$ of the original objective costs.}
Additionally, $PPSM_d$ is consistently more accurate than $PPSM_p$. The reason is that, by enforcing fidelity of the coordination variables $\bm{y}^g$, the $PPSM_d$ better limits the impact of the noisy data on the leader objective (GAUC), which in turn results in more faithful solutions for the followers problems.

\begin{table}
\centering
\begin{tabular}{llrrr}
\toprule
     $\cM$ & $\eta$ & $\Delta_{\cO^{uc}} (\%)$ &   $\Delta_{\cO^e} (\%)$&   $\Delta_{\cO^g} (\%)$ \\
\midrule
\multirow{1}{*}{Laplace}
     & NA &  21.0501 &  43.1685 &  43.1685 \\
\cline{2-5}
\multirow{3}{*}{$\textsl{PPSM}_{p}$}
     & 0.1\% &   0.1152 &   0.2136 &   0.2136 \\
     & 1.0\% &   0.1173 &   0.2223 &   0.2222 \\
     & 10.0\%&   0.1380 &   0.2567 &   0.2567 \\
\cline{2-5}
\multirow{3}{*}{$\textsl{PPSM}_{d}$}
     & 0.1\% &   0.1109 &   0.2223 &   0.2223 \\
     & 1.0\% &   0.1142 &   0.2080 &   0.2079 \\
     & 10.0\% &  0.1331 &   0.2469 &   0.2469 \\
     \bottomrule
\end{tabular}
\caption{Cost objective differences (\%) at varying fidelity parameters $\eta$ \%, and indistinguishability parameter $\alpha=10$.
\label{tab:2}}
\end{table}
These results are further illustrated in Table \ref{tab:2}, that analyzes the difference in objective costs at varying fidelity parameters $\eta$, for a fixed indistinguishability parameter $\alpha = 10$ (i.e., the largest privacy level attainable in our experimental setting). Once again, the results of the PPSM mechanisms are at least two orders of magnitude more precise than those obtained by the Laplace mechanism. Additionally, notice that the fidelity parameter $\eta$ impacts the accuracy of the privacy-preserving objective costs. Indeed, the fidelity parameter indirectly controls the deviation of the privacy-preserving GAUC and GM objectives with respect to the original ones. While the results differences are small, in percentage, their impact on the objective functions (which are in the order of $10^6$) is non-negligible.

\begin{figure*}
\includegraphics[width=0.33\linewidth]{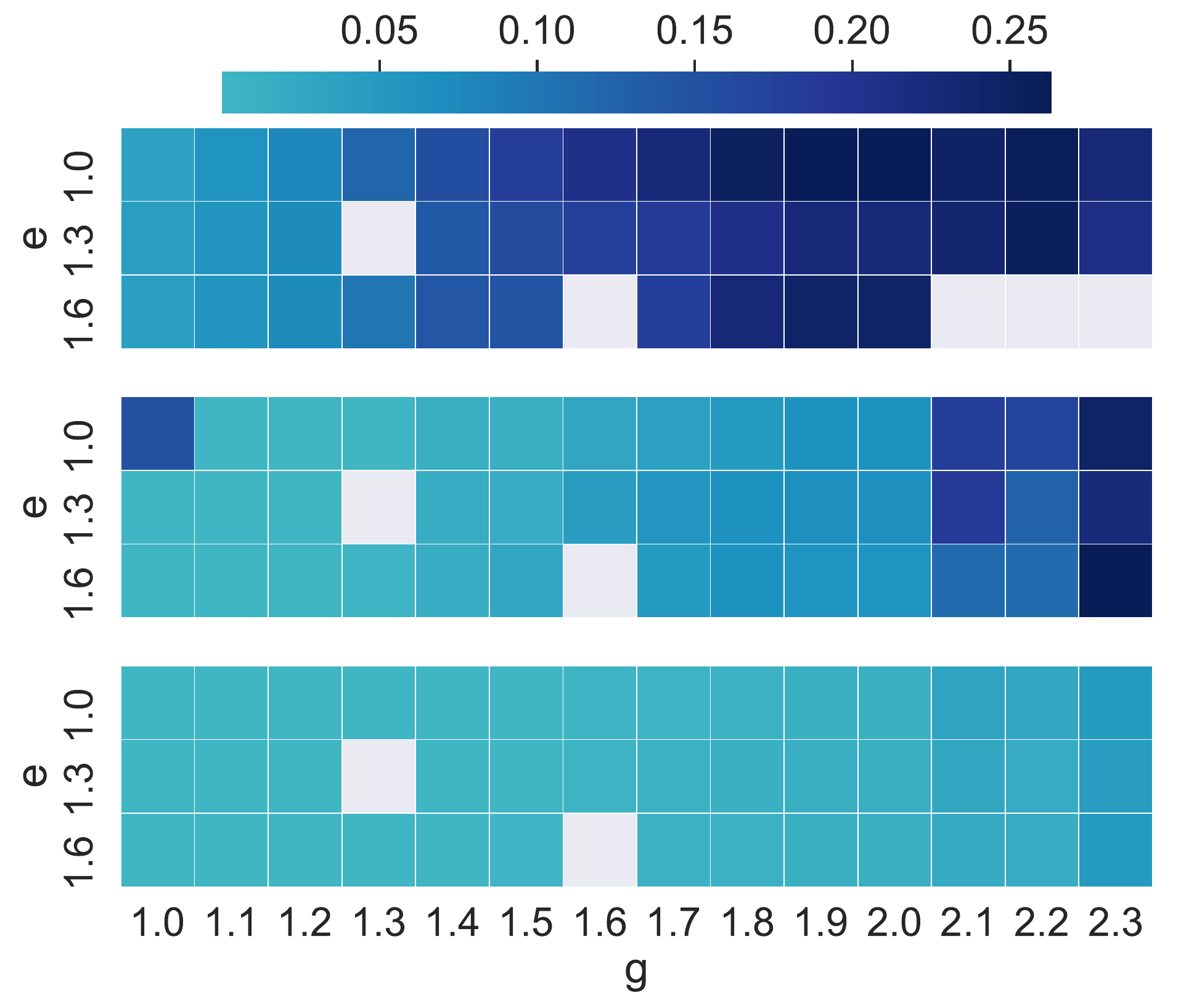}
\includegraphics[width=0.33\linewidth]{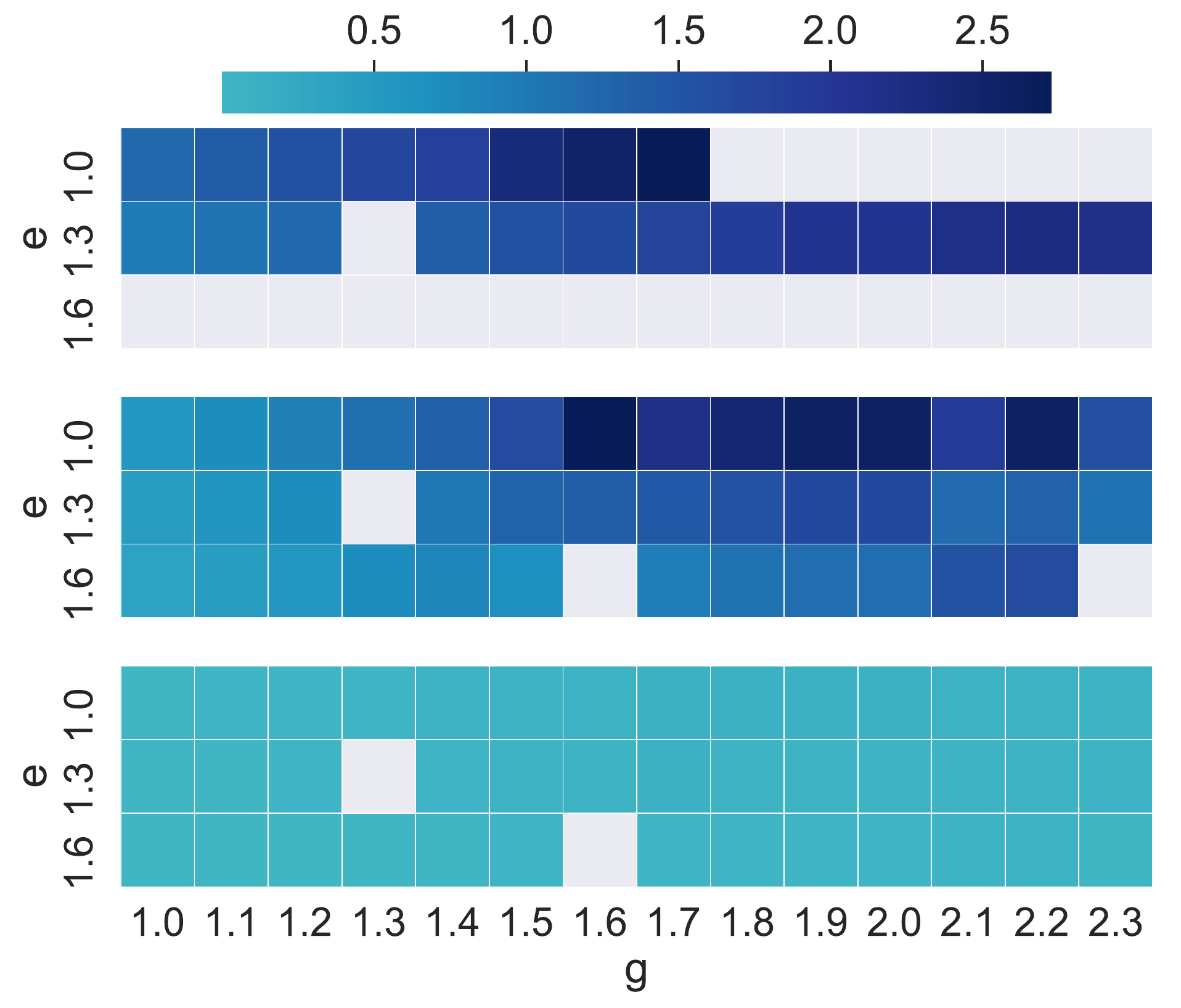}
\includegraphics[width=0.33\linewidth]{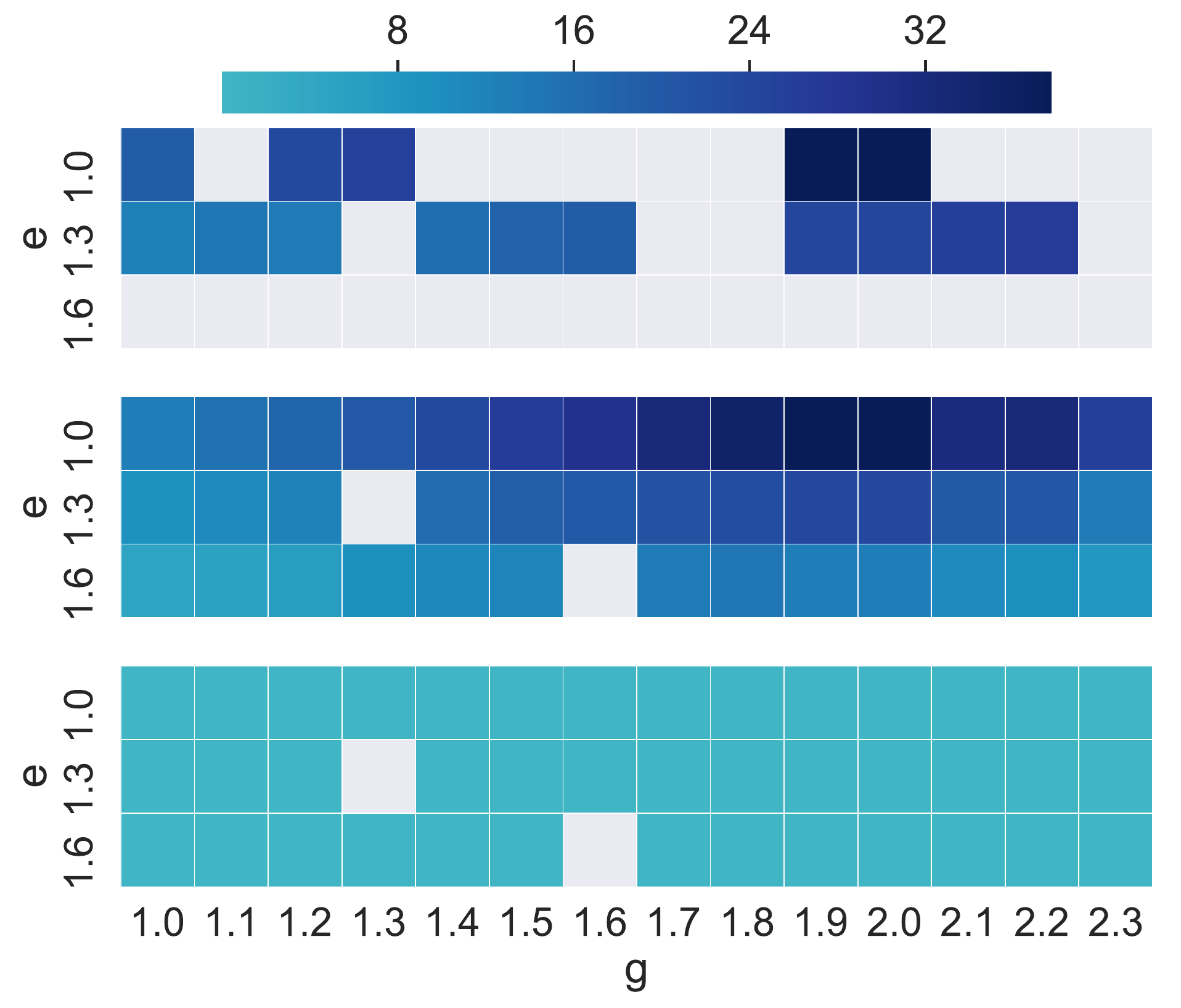}

\caption{Total (GAUC and GM) objective cost difference (\%) at varying gas (g) and electricity (e) stress levels for privacy-preserving data obtained via \emph{Laplace Mechanism} (top), \emph{$PPSM_p$} (middle), \emph{$PPSM_d$} (bottom). Indistinguishably parameters: Left: $\alpha=0.1$, Center: $\alpha=1.0$, Right: $\alpha=10.0$. Fidelity value $\eta=0.001 \%$.
\label{fig:3}
}
\end{figure*}

\subsection{Stress Levels Analysis}
Finally, this subsection details the effect of the privacy preserving mechanisms on the combined GAUC problem for all the electricity and gas stress levels adopted. 

Figure \ref{fig:3} reports heatmaps of the total (GAUC and GM) objective cost difference, in percentage, at varying electricity (e) and gas (g) stress levels for the privacy preserving data obtained via the Laplace mechanism (top), $PPSM_p$ (middle), and $PPSM_d$ (bottom). Each square represents the objective cost difference averaged over $30$ instances for a particular electricity and gas stress level. The darker the color, the more pronounced are the errors committed by the mechanisms, as reported in the legends on top of each subfigure. Gray squares represent the set of instances for which no feasible solution of the GAUC problem was found or when a timeout is reached. 
The illustration reports the cost differences for indistinguishability parameters $\alpha=0.1$ (left subfigure) $\alpha=1.0$ (middle subfigure), and $\alpha=10.0$ (right subfigure).

These results illustrate three trends: 
Firstly, they show that, for every mechanism, the objective differences becomes more pronounced as the electricity and gas stress levels increase. This can be explained by the increased impact of the Laplace perturbations on higher values of gas demand profiles.
Secondly, they remark that the PPSMs produce privacy-preserving Stackelberg problems which are consistently more faithful to the original problems with respect to those produced by the Laplace mechanism.
Finally, they show that $PPSM_d$ is consistently more accurate than $PPSM_p$ over all stress levels. \emph{These results are significant, as they show the robustness of the proposed PPSMs over different electricity and natural gas demand profiles. They indicate that these PPSMs can provide a realistic and efficient solution for the coordination of electricity and natural gas markets.}





\section{Related work}
\label{sec:related_work}

The obfuscation of data values under the lens of differential privacy is a challenging task that has been studied from several angles. Often, the released data is generated from a data synopsis in the form of a noisy histogram \cite{li2010optimizing,hay2016principled,qardaji:14,kasiviswanathan2010price,mohammed:11,xiao2010differentially}. These methods are typically adopted in the context of statistical queries. 

The design of markets for private data has also received considerable attention. For example, Ghosh and Roth 
design a truthful mechanism that allows a data analyst to buy information from a population from which they can estimate some statistic \cite{ghosh2010selling}. Niu et al.~study mechanisms to trade noisy aggregate statistics from the perspective of a data broker in data markets \cite{Niu:2018}. 

However, all the proposals above, do not involve data used as input to a complex optimization problem, as in the case of this work. 
The closest work related to the proposal in this paper can be considered~\cite{Fioretto:18b}, that, in the context energy networks, propose a privacy-preserving mechanism for releasing datasets that can be used as input to an \emph{optimal power flow} problem. A similar line of work uses hierarchical (bilevel) optimization for obfuscating the energy network parameters or locations while ensuring high utility for the problem of interest \cite{Fioretto:ijcai-19,Fioretto:TSG19}. 

In contrast to the aforementioned studies, the proposed PPSM focuses on solving Stackelberg games in which the followers parameters are sensitive. The PPSM also enforces two notions of fidelity of the privacy-preserving information to the leader and/or follower objectives. Finally, to the best of our knowledge, this is the first differentially-private mechanism that is applied to the coordination of sequential electricity and natural gas markets.  

\section{Conclusions}
\label{sec:conclusions}

This paper introduced a differentially private mechanism to protect the \emph{sensitive information} exchanged during the coordination of the sequential electricity and natural gas market clearings. The \emph{proposed Privacy-Preserving Stackelberg Mechanism (PPSM)} obfuscates the gas demand profile exchanged by the gas market, while also ensuring that the resulting problem preserves the fundamental properties of the original Stackelberg game. Specifically,
the PPSM was shown to enjoy strong properties: It complies with the notion of differential privacy and ensures that the outcomes of the privacy-preserving Stackelberg mechanism are close-to-optimality for each agent.
Experimental results on several gas and electricity market benchmarks based on a real case study demonstrated the effectiveness of the approach: \emph{The PPSM was shown to obtain up to two orders of magnitude improvement on the cost of the agents when compared to a traditional Laplace mechanism}.

While the proposed methods were detailed in the context of a sequential and independent natural gas and electricity market, they open up various domains of application. For example, the market-based coordination 
approach for electricity and district heating systems proposed in 
\cite{mitridati2016optimal,mitridati2019bid}, the game-theoretical coordination of transmission and distribution systems developed in \cite{yuan2017hierarchical,le2019game,hassan2018energy}, or the real-time pricing scheme for coordinating the demand response from aggregators and consumers \cite{zhang2016real,momber2015retail}.

Future work will focus on several avenues, including extended theoretical bounds on the cost of privacy, studying the game-theoretic properties of this privacy-preserving Stackelberg game, accounting for uncertainty on the public data, and, studying the applicability of the PPSM to other domains.


\bibliographystyle{ACM-Reference-Format}  
\bibliography{bibliography,aamas}  


\begin{thebibliography}{37}


\ifx \showCODEN    \undefined \def \showCODEN     #1{\unskip}     \fi
\ifx \showDOI      \undefined \def \showDOI       #1{#1}\fi
\ifx \showISBNx    \undefined \def \showISBNx     #1{\unskip}     \fi
\ifx \showISBNxiii \undefined \def \showISBNxiii  #1{\unskip}     \fi
\ifx \showISSN     \undefined \def \showISSN      #1{\unskip}     \fi
\ifx \showLCCN     \undefined \def \showLCCN      #1{\unskip}     \fi
\ifx \shownote     \undefined \def \shownote      #1{#1}          \fi
\ifx \showarticletitle \undefined \def \showarticletitle #1{#1}   \fi
\ifx \showURL      \undefined \def \showURL       {\relax}        \fi
\providecommand\bibfield[2]{#2}
\providecommand\bibinfo[2]{#2}
\providecommand\natexlab[1]{#1}
\providecommand\showeprint[2][]{arXiv:#2}

\bibitem[\protect\citeauthoryear{Alabdulwahab, Abusorrah, Zhang, and
  Shahidehpour}{Alabdulwahab et~al\mbox{.}}{2015}]%
        {alabdulwahab2015coordination}
\bibfield{author}{\bibinfo{person}{Ahmed Alabdulwahab},
  \bibinfo{person}{Abdullah Abusorrah}, \bibinfo{person}{Xiaping Zhang}, {and}
  \bibinfo{person}{Mohammad Shahidehpour}.} \bibinfo{year}{2015}\natexlab{}.
\newblock \showarticletitle{Coordination of interdependent natural gas and
  electricity infrastructures for firming the variability of wind energy in
  stochastic day-ahead scheduling}.
\newblock \bibinfo{journal}{\emph{IEEE Transactions on Sustainable Energy}}
  \bibinfo{volume}{6}, \bibinfo{number}{2} (\bibinfo{year}{2015}),
  \bibinfo{pages}{606--615}.
\newblock


\bibitem[\protect\citeauthoryear{Allen, Lang, and Ilic}{Allen
  et~al\mbox{.}}{2008}]%
        {allen2008combined}
\bibfield{author}{\bibinfo{person}{Eric~H Allen}, \bibinfo{person}{Jeffrey~H
  Lang}, {and} \bibinfo{person}{Marija~D Ilic}.}
  \bibinfo{year}{2008}\natexlab{}.
\newblock \showarticletitle{A combined equivalenced-electric, economic, and
  market representation of the northeastern power coordinating council us
  electric power system}.
\newblock \bibinfo{journal}{\emph{IEEE Transactions on Power Systems}}
  \bibinfo{volume}{23}, \bibinfo{number}{3} (\bibinfo{year}{2008}),
  \bibinfo{pages}{896--907}.
\newblock


\bibitem[\protect\citeauthoryear{Andr{\'e}s, Bordenabe, Chatzikokolakis, and
  Palamidessi}{Andr{\'e}s et~al\mbox{.}}{2013}]%
        {andres2013geo}
\bibfield{author}{\bibinfo{person}{Miguel~E Andr{\'e}s},
  \bibinfo{person}{Nicol{\'a}s~E Bordenabe}, \bibinfo{person}{Konstantinos
  Chatzikokolakis}, {and} \bibinfo{person}{Catuscia Palamidessi}.}
  \bibinfo{year}{2013}\natexlab{}.
\newblock \showarticletitle{Geo-indistinguishability: Differential privacy for
  location-based systems}. In \bibinfo{booktitle}{\emph{Proceedings of the 2013
  ACM SIGSAC conference on Computer \& communications security}}. ACM,
  \bibinfo{pages}{901--914}.
\newblock


\bibitem[\protect\citeauthoryear{Anonymous}{Anonymous}{2019}]%
        {onlineappendix}
\bibfield{author}{\bibinfo{person}{Anonymous}.}
  \bibinfo{year}{2019}\natexlab{}.
\newblock \bibinfo{title}{PPSM: A Privacy-Preserving Stackelberg Mechanism
  (Online Appendix)}.
\newblock \bibinfo{howpublished}{Online}.
\newblock
\urldef\tempurl%
\url{{https://tinyurl.com/yekqpzgu}}
\showURL{%
Retrieved Nov 15, 2019 from \tempurl}


\bibitem[\protect\citeauthoryear{Apple}{Apple}{2017}]%
        {apple}
\bibfield{author}{\bibinfo{person}{Apple}.} \bibinfo{year}{2017}\natexlab{}.
\newblock \bibinfo{title}{Differential Privacy Overview}.
\newblock
  \bibinfo{howpublished}{\url{https://images.apple.com/privacy/docs/Differential_Privacy_Overview.pdf}}.
\newblock


\bibitem[\protect\citeauthoryear{Bent, Blumsack, Van~Hentenryck,
  Borraz-S{\'a}nchez, and Shahriari}{Bent et~al\mbox{.}}{2018}]%
        {bent2018joint}
\bibfield{author}{\bibinfo{person}{Russell Bent}, \bibinfo{person}{Seth
  Blumsack}, \bibinfo{person}{Pascal Van~Hentenryck}, \bibinfo{person}{Conrado
  Borraz-S{\'a}nchez}, {and} \bibinfo{person}{Mehdi Shahriari}.}
  \bibinfo{year}{2018}\natexlab{}.
\newblock \showarticletitle{Joint electricity and natural gas transmission
  planning with endogenous market feedbacks}.
\newblock \bibinfo{journal}{\emph{IEEE Transactions on Power Systems}}
  \bibinfo{volume}{33}, \bibinfo{number}{6} (\bibinfo{year}{2018}),
  \bibinfo{pages}{6397--6409}.
\newblock


\bibitem[\protect\citeauthoryear{Byeon and Van~Hentenryck}{Byeon and
  Van~Hentenryck}{2019}]%
        {byeon2019unit}
\bibfield{author}{\bibinfo{person}{Geunyeong Byeon} {and}
  \bibinfo{person}{Pascal Van~Hentenryck}.} \bibinfo{year}{2019}\natexlab{}.
\newblock \showarticletitle{Unit Commitment With Gas Network Awareness}.
\newblock \bibinfo{journal}{\emph{arXiv preprint arXiv:1902.03236}}
  (\bibinfo{year}{2019}).
\newblock


\bibitem[\protect\citeauthoryear{Chatzikokolakis, Andr{\'e}s, Bordenabe, and
  Palamidessi}{Chatzikokolakis et~al\mbox{.}}{2013}]%
        {chatzikokolakis2013broadening}
\bibfield{author}{\bibinfo{person}{Konstantinos Chatzikokolakis},
  \bibinfo{person}{Miguel~E Andr{\'e}s}, \bibinfo{person}{Nicol{\'a}s~Emilio
  Bordenabe}, {and} \bibinfo{person}{Catuscia Palamidessi}.}
  \bibinfo{year}{2013}\natexlab{}.
\newblock \showarticletitle{Broadening the scope of differential privacy using
  metrics}. In \bibinfo{booktitle}{\emph{International Symposium on Privacy
  Enhancing Technologies Symposium}}. Springer, \bibinfo{pages}{82--102}.
\newblock


\bibitem[\protect\citeauthoryear{Dwork, McSherry, Nissim, and Smith}{Dwork
  et~al\mbox{.}}{2006}]%
        {dwork:06}
\bibfield{author}{\bibinfo{person}{Cynthia Dwork}, \bibinfo{person}{Frank
  McSherry}, \bibinfo{person}{Kobbi Nissim}, {and} \bibinfo{person}{Adam
  Smith}.} \bibinfo{year}{2006}\natexlab{}.
\newblock \showarticletitle{Calibrating noise to sensitivity in private data
  analysis}. In \bibinfo{booktitle}{\emph{TCC}}, Vol.~\bibinfo{volume}{3876}.
  Springer, \bibinfo{pages}{265--284}.
\newblock


\bibitem[\protect\citeauthoryear{Dwork and Roth}{Dwork and Roth}{2013}]%
        {dwork:13}
\bibfield{author}{\bibinfo{person}{Cynthia Dwork} {and} \bibinfo{person}{Aaron
  Roth}.} \bibinfo{year}{2013}\natexlab{}.
\newblock \showarticletitle{The algorithmic foundations of differential
  privacy}.
\newblock \bibinfo{journal}{\emph{Theoretical Computer Science}}
  \bibinfo{volume}{9}, \bibinfo{number}{3-4} (\bibinfo{year}{2013}),
  \bibinfo{pages}{211--407}.
\newblock


\bibitem[\protect\citeauthoryear{Fanti, Pihur, and Erlingsson}{Fanti
  et~al\mbox{.}}{2016}]%
        {fanti:16}
\bibfield{author}{\bibinfo{person}{Giulia Fanti}, \bibinfo{person}{Vasyl
  Pihur}, {and} \bibinfo{person}{{\'U}lfar Erlingsson}.}
  \bibinfo{year}{2016}\natexlab{}.
\newblock \showarticletitle{Building a RAPPOR with the unknown:
  Privacy-preserving learning of associations and data dictionaries}.
\newblock \bibinfo{journal}{\emph{Proceedings on Privacy Enhancing
  Technologies}} \bibinfo{volume}{2016}, \bibinfo{number}{3}
  (\bibinfo{year}{2016}), \bibinfo{pages}{41--61}.
\newblock


\bibitem[\protect\citeauthoryear{Fioretto and Hentenryck}{Fioretto and
  Hentenryck}{2018}]%
        {fioretto:CPAIOR-18}
\bibfield{author}{\bibinfo{person}{Ferdinando Fioretto} {and}
  \bibinfo{person}{Pascal~Van Hentenryck}.} \bibinfo{year}{2018}\natexlab{}.
\newblock \showarticletitle{Constrained-based Differential Privacy: Releasing
  Optimal Power Flow Benchmarks Privately}. In
  \bibinfo{booktitle}{\emph{Proceedings of the International Conference on the
  Integration of Constraint Programming, Artificial Intelligence, and
  Operations Research {(CPAIOR)}}}. \bibinfo{pages}{215--231}.
\newblock


\bibitem[\protect\citeauthoryear{Fioretto, Mak, and {Van Hentenryck}}{Fioretto
  et~al\mbox{.}}{2019a}]%
        {Fioretto:TSG19}
\bibfield{author}{\bibinfo{person}{Ferdinando Fioretto},
  \bibinfo{person}{{Terrence W.K.} Mak}, {and} \bibinfo{person}{Pascal {Van
  Hentenryck}}.} \bibinfo{year}{2019}\natexlab{a}.
\newblock \showarticletitle{Differential Privacy for Power Grid Obfuscation}.
\newblock \bibinfo{journal}{\emph{IEEE Transactions on Smart Grid}}
  (\bibinfo{year}{2019}), \bibinfo{pages}{1--1}.
\newblock
\urldef\tempurl%
\url{https://doi.org/10.1109/TSG.2019.2936712}
\showDOI{\tempurl}


\bibitem[\protect\citeauthoryear{Fioretto, Mak, and {Van Hentenryck}}{Fioretto
  et~al\mbox{.}}{2019b}]%
        {Fioretto:ijcai-19}
\bibfield{author}{\bibinfo{person}{Ferdinando Fioretto},
  \bibinfo{person}{Terrence W.~K. Mak}, {and} \bibinfo{person}{Pascal {Van
  Hentenryck}}.} \bibinfo{year}{2019}\natexlab{b}.
\newblock \showarticletitle{Privacy-Preserving Obfuscation of Critical
  Infrastructure Networks}.
\newblock  (\bibinfo{year}{2019}), \bibinfo{pages}{1086--1092}.
\newblock


\bibitem[\protect\citeauthoryear{Fioretto and Van~Hentenryck}{Fioretto and
  Van~Hentenryck}{2018}]%
        {Fioretto:18b}
\bibfield{author}{\bibinfo{person}{Ferdinando Fioretto} {and}
  \bibinfo{person}{Pascal Van~Hentenryck}.} \bibinfo{year}{2018}\natexlab{}.
\newblock \showarticletitle{Constrained-Based Differential Privacy: Releasing
  Optimal Power Flow Benchmarks Privately}. In
  \bibinfo{booktitle}{\emph{Proceedings of Integration of Constraint
  Programming, Artificial Intelligence, and Operations Research (CPAIOR)}}.
  \bibinfo{pages}{215--231}.
\newblock


\bibitem[\protect\citeauthoryear{Gabriel, Conejo, Fuller, Hobbs, and
  Ruiz}{Gabriel et~al\mbox{.}}{2012}]%
        {gabriel2012complementarity}
\bibfield{author}{\bibinfo{person}{Steven~A Gabriel},
  \bibinfo{person}{Antonio~J Conejo}, \bibinfo{person}{J~David Fuller},
  \bibinfo{person}{Benjamin~F Hobbs}, {and} \bibinfo{person}{Carlos Ruiz}.}
  \bibinfo{year}{2012}\natexlab{}.
\newblock \bibinfo{booktitle}{\emph{Complementarity modeling in energy
  markets}}. Vol.~\bibinfo{volume}{180}.
\newblock \bibinfo{publisher}{Springer Science \& Business Media}.
\newblock


\bibitem[\protect\citeauthoryear{Ghosh and Roth}{Ghosh and Roth}{2010}]%
        {ghosh2010selling}
\bibfield{author}{\bibinfo{person}{Arpita Ghosh} {and} \bibinfo{person}{Aaron
  Roth}.} \bibinfo{year}{2010}\natexlab{}.
\newblock \bibinfo{title}{Selling Privacy at Auction}.
\newblock
\newblock
\showeprint[arxiv]{cs.GT/1011.1375}


\bibitem[\protect\citeauthoryear{Hassan and Dvorkin}{Hassan and
  Dvorkin}{2018}]%
        {hassan2018energy}
\bibfield{author}{\bibinfo{person}{Ali Hassan} {and} \bibinfo{person}{Yury
  Dvorkin}.} \bibinfo{year}{2018}\natexlab{}.
\newblock \showarticletitle{Energy storage siting and sizing in coordinated
  distribution and transmission systems}.
\newblock \bibinfo{journal}{\emph{IEEE Transactions on Sustainable Energy}}
  \bibinfo{volume}{9}, \bibinfo{number}{4} (\bibinfo{year}{2018}),
  \bibinfo{pages}{1692--1701}.
\newblock


\bibitem[\protect\citeauthoryear{Hay, Machanavajjhala, Miklau, Chen, and
  Zhang}{Hay et~al\mbox{.}}{2016}]%
        {hay2016principled}
\bibfield{author}{\bibinfo{person}{Michael Hay}, \bibinfo{person}{Ashwin
  Machanavajjhala}, \bibinfo{person}{Gerome Miklau}, \bibinfo{person}{Yan
  Chen}, {and} \bibinfo{person}{Dan Zhang}.} \bibinfo{year}{2016}\natexlab{}.
\newblock \showarticletitle{Principled evaluation of differentially private
  algorithms using dpbench}. In \bibinfo{booktitle}{\emph{Proceedings of the
  2016 International Conference on Management of Data}}. ACM,
  \bibinfo{pages}{139--154}.
\newblock


\bibitem[\protect\citeauthoryear{Kasiviswanathan, Rudelson, Smith, and
  Ullman}{Kasiviswanathan et~al\mbox{.}}{2010}]%
        {kasiviswanathan2010price}
\bibfield{author}{\bibinfo{person}{Shiva~Prasad Kasiviswanathan},
  \bibinfo{person}{Mark Rudelson}, \bibinfo{person}{Adam Smith}, {and}
  \bibinfo{person}{Jonathan Ullman}.} \bibinfo{year}{2010}\natexlab{}.
\newblock \showarticletitle{The price of privately releasing contingency tables
  and the spectra of random matrices with correlated rows}. In
  \bibinfo{booktitle}{\emph{Proceedings of the forty-second ACM symposium on
  Theory of computing}}. ACM, \bibinfo{pages}{775--784}.
\newblock


\bibitem[\protect\citeauthoryear{Koufogiannis, Han, and Pappas}{Koufogiannis
  et~al\mbox{.}}{2015}]%
        {koufogiannis:15}
\bibfield{author}{\bibinfo{person}{Fragkiskos Koufogiannis},
  \bibinfo{person}{Shuo Han}, {and} \bibinfo{person}{George~J Pappas}.}
  \bibinfo{year}{2015}\natexlab{}.
\newblock \showarticletitle{Optimality of the laplace mechanism in differential
  privacy}.
\newblock \bibinfo{journal}{\emph{arXiv preprint arXiv:1504.00065}}
  (\bibinfo{year}{2015}).
\newblock


\bibitem[\protect\citeauthoryear{Krall, Higgins, and O’Neill}{Krall
  et~al\mbox{.}}{2012}]%
        {krall2012rto}
\bibfield{author}{\bibinfo{person}{Eric Krall}, \bibinfo{person}{Michael
  Higgins}, {and} \bibinfo{person}{Richard~P O’Neill}.}
  \bibinfo{year}{2012}\natexlab{}.
\newblock \showarticletitle{RTO unit commitment test system}.
\newblock \bibinfo{journal}{\emph{Federal Energy Regulatory Commission}}
  (\bibinfo{year}{2012}).
\newblock


\bibitem[\protect\citeauthoryear{Le~Cadre, Mezghani, and Papavasiliou}{Le~Cadre
  et~al\mbox{.}}{2019}]%
        {le2019game}
\bibfield{author}{\bibinfo{person}{H{\'e}l{\`e}ne Le~Cadre},
  \bibinfo{person}{Ily{\`e}s Mezghani}, {and} \bibinfo{person}{Anthony
  Papavasiliou}.} \bibinfo{year}{2019}\natexlab{}.
\newblock \showarticletitle{A game-theoretic analysis of
  transmission-distribution system operator coordination}.
\newblock \bibinfo{journal}{\emph{European Journal of Operational Research}}
  \bibinfo{volume}{274}, \bibinfo{number}{1} (\bibinfo{year}{2019}),
  \bibinfo{pages}{317--339}.
\newblock


\bibitem[\protect\citeauthoryear{Li, Hay, Rastogi, Miklau, and McGregor}{Li
  et~al\mbox{.}}{2010}]%
        {li2010optimizing}
\bibfield{author}{\bibinfo{person}{Chao Li}, \bibinfo{person}{Michael Hay},
  \bibinfo{person}{Vibhor Rastogi}, \bibinfo{person}{Gerome Miklau}, {and}
  \bibinfo{person}{Andrew McGregor}.} \bibinfo{year}{2010}\natexlab{}.
\newblock \showarticletitle{Optimizing linear counting queries under
  differential privacy}. In \bibinfo{booktitle}{\emph{Proceedings of the
  twenty-ninth ACM SIGMOD-SIGACT-SIGART symposium on Principles of database
  systems}}. ACM, \bibinfo{pages}{123--134}.
\newblock


\bibitem[\protect\citeauthoryear{Maharjan, Zhu, Zhang, Gjessing, and
  Basar}{Maharjan et~al\mbox{.}}{2013}]%
        {maharjan2013dependable}
\bibfield{author}{\bibinfo{person}{Sabita Maharjan}, \bibinfo{person}{Quanyan
  Zhu}, \bibinfo{person}{Yan Zhang}, \bibinfo{person}{Stein Gjessing}, {and}
  \bibinfo{person}{Tamer Basar}.} \bibinfo{year}{2013}\natexlab{}.
\newblock \showarticletitle{Dependable demand response management in the smart
  grid: A Stackelberg game approach}.
\newblock \bibinfo{journal}{\emph{IEEE Transactions on Smart Grid}}
  \bibinfo{volume}{4}, \bibinfo{number}{1} (\bibinfo{year}{2013}),
  \bibinfo{pages}{120--132}.
\newblock


\bibitem[\protect\citeauthoryear{Mitridati and Pinson}{Mitridati and
  Pinson}{2016}]%
        {mitridati2016optimal}
\bibfield{author}{\bibinfo{person}{Lesia Mitridati} {and}
  \bibinfo{person}{Pierre Pinson}.} \bibinfo{year}{2016}\natexlab{}.
\newblock \showarticletitle{Optimal coupling of heat and electricity systems: A
  stochastic hierarchical approach}. In \bibinfo{booktitle}{\emph{2016
  International Conference on Probabilistic Methods Applied to Power Systems
  (PMAPS)}}. IEEE, \bibinfo{pages}{1--6}.
\newblock


\bibitem[\protect\citeauthoryear{Mitridati, Van~Hentenryck,
  et~al\mbox{.}}{Mitridati et~al\mbox{.}}{2019}]%
        {mitridati2019bid}
\bibfield{author}{\bibinfo{person}{Lesia Mitridati}, \bibinfo{person}{Pascal
  Van~Hentenryck}, {et~al\mbox{.}}} \bibinfo{year}{2019}\natexlab{}.
\newblock \showarticletitle{A Bid-Validity Mechanism for Sequential Heat and
  Electricity Market Clearing}.
\newblock \bibinfo{journal}{\emph{arXiv preprint arXiv:1910.08617}}
  (\bibinfo{year}{2019}).
\newblock


\bibitem[\protect\citeauthoryear{Mohammed, Chen, Fung, and Yu}{Mohammed
  et~al\mbox{.}}{2011}]%
        {mohammed:11}
\bibfield{author}{\bibinfo{person}{Noman Mohammed}, \bibinfo{person}{Rui Chen},
  \bibinfo{person}{Benjamin Fung}, {and} \bibinfo{person}{Philip~S Yu}.}
  \bibinfo{year}{2011}\natexlab{}.
\newblock \showarticletitle{Differentially private data release for data
  mining}. In \bibinfo{booktitle}{\emph{Proceedings of the 17th ACM SIGKDD
  international conference on Knowledge discovery and data mining}}. ACM,
  \bibinfo{pages}{493--501}.
\newblock


\bibitem[\protect\citeauthoryear{Momber, Wogrin, and San~Rom{\'a}n}{Momber
  et~al\mbox{.}}{2015}]%
        {momber2015retail}
\bibfield{author}{\bibinfo{person}{Ilan Momber}, \bibinfo{person}{Sonja
  Wogrin}, {and} \bibinfo{person}{Tom{\'a}s~G{\'o}mez San~Rom{\'a}n}.}
  \bibinfo{year}{2015}\natexlab{}.
\newblock \showarticletitle{Retail pricing: A bilevel program for PEV
  aggregator decisions using indirect load control}.
\newblock \bibinfo{journal}{\emph{IEEE Transactions on Power Systems}}
  \bibinfo{volume}{31}, \bibinfo{number}{1} (\bibinfo{year}{2015}),
  \bibinfo{pages}{464--473}.
\newblock


\bibitem[\protect\citeauthoryear{Niu, Zheng, Wu, Tang, Gao, and Chen}{Niu
  et~al\mbox{.}}{2018}]%
        {Niu:2018}
\bibfield{author}{\bibinfo{person}{Chaoyue Niu}, \bibinfo{person}{Zhenzhe
  Zheng}, \bibinfo{person}{Fan Wu}, \bibinfo{person}{Shaojie Tang},
  \bibinfo{person}{Xiaofeng Gao}, {and} \bibinfo{person}{Guihai Chen}.}
  \bibinfo{year}{2018}\natexlab{}.
\newblock \showarticletitle{Unlocking the Value of Privacy: Trading Aggregate
  Statistics over Private Correlated Data}. In
  \bibinfo{booktitle}{\emph{Proceedings of the 24th ACM SIGKDD International
  Conference on Knowledge Discovery \&\#38; Data Mining}}
  \emph{(\bibinfo{series}{KDD '18})}. \bibinfo{publisher}{ACM},
  \bibinfo{address}{New York, NY, USA}, \bibinfo{pages}{2031--2040}.
\newblock
\showISBNx{978-1-4503-5552-0}
\urldef\tempurl%
\url{https://doi.org/10.1145/3219819.3220013}
\showDOI{\tempurl}


\bibitem[\protect\citeauthoryear{Ordoudis}{Ordoudis}{2018}]%
        {ordoudis2018market}
\bibfield{author}{\bibinfo{person}{Christos Ordoudis}.}
  \bibinfo{year}{2018}\natexlab{}.
\newblock \showarticletitle{Market-based Approaches for the Coordinated
  Operation of Electricity and Natural Gas Systems}.
\newblock  (\bibinfo{year}{2018}).
\newblock


\bibitem[\protect\citeauthoryear{Pinson, Mitridati, Ordoudis, and
  Ostergaard}{Pinson et~al\mbox{.}}{2017}]%
        {pinson2017towards}
\bibfield{author}{\bibinfo{person}{Pierre Pinson}, \bibinfo{person}{Lesia
  Mitridati}, \bibinfo{person}{Christos Ordoudis}, {and} \bibinfo{person}{Jacob
  Ostergaard}.} \bibinfo{year}{2017}\natexlab{}.
\newblock \showarticletitle{Towards fully renewable energy systems: Experience
  and trends in Denmark}.
\newblock \bibinfo{journal}{\emph{CSEE journal of power and energy systems}}
  \bibinfo{volume}{3}, \bibinfo{number}{1} (\bibinfo{year}{2017}),
  \bibinfo{pages}{26--35}.
\newblock


\bibitem[\protect\citeauthoryear{Qardaji, Yang, and Li}{Qardaji
  et~al\mbox{.}}{2014}]%
        {qardaji:14}
\bibfield{author}{\bibinfo{person}{Wahbeh Qardaji}, \bibinfo{person}{Weining
  Yang}, {and} \bibinfo{person}{Ninghui Li}.} \bibinfo{year}{2014}\natexlab{}.
\newblock \showarticletitle{PriView: practical differentially private release
  of marginal contingency tables}. In \bibinfo{booktitle}{\emph{Proceedings of
  the 2014 ACM SIGMOD international conference on Management of data}}. ACM,
  \bibinfo{pages}{1435--1446}.
\newblock


\bibitem[\protect\citeauthoryear{Simaan and Cruz}{Simaan and Cruz}{1973}]%
        {simaan1973stackelberg}
\bibfield{author}{\bibinfo{person}{Marwaan Simaan} {and}
  \bibinfo{person}{Jose~B Cruz}.} \bibinfo{year}{1973}\natexlab{}.
\newblock \showarticletitle{On the Stackelberg strategy in nonzero-sum games}.
\newblock \bibinfo{journal}{\emph{Journal of Optimization Theory and
  Applications}} \bibinfo{volume}{11}, \bibinfo{number}{5}
  (\bibinfo{year}{1973}), \bibinfo{pages}{533--555}.
\newblock


\bibitem[\protect\citeauthoryear{Xiao, Xiong, and Yuan}{Xiao
  et~al\mbox{.}}{2010}]%
        {xiao2010differentially}
\bibfield{author}{\bibinfo{person}{Yonghui Xiao}, \bibinfo{person}{Li Xiong},
  {and} \bibinfo{person}{Chun Yuan}.} \bibinfo{year}{2010}\natexlab{}.
\newblock \showarticletitle{Differentially private data release through
  multidimensional partitioning}. In \bibinfo{booktitle}{\emph{Workshop on
  Secure Data Management}}. Springer, \bibinfo{pages}{150--168}.
\newblock


\bibitem[\protect\citeauthoryear{Yuan and Hesamzadeh}{Yuan and
  Hesamzadeh}{2017}]%
        {yuan2017hierarchical}
\bibfield{author}{\bibinfo{person}{Zhao Yuan} {and}
  \bibinfo{person}{Mohammad~Reza Hesamzadeh}.} \bibinfo{year}{2017}\natexlab{}.
\newblock \showarticletitle{Hierarchical coordination of TSO-DSO economic
  dispatch considering large-scale integration of distributed energy
  resources}.
\newblock \bibinfo{journal}{\emph{Applied energy}}  \bibinfo{volume}{195}
  (\bibinfo{year}{2017}), \bibinfo{pages}{600--615}.
\newblock


\bibitem[\protect\citeauthoryear{Zhang, Wang, Wang, Pinson, Morales, and
  {\O}stergaard}{Zhang et~al\mbox{.}}{2016}]%
        {zhang2016real}
\bibfield{author}{\bibinfo{person}{Chunyu Zhang}, \bibinfo{person}{Qi Wang},
  \bibinfo{person}{Jianhui Wang}, \bibinfo{person}{Pierre Pinson},
  \bibinfo{person}{Juan~M Morales}, {and} \bibinfo{person}{Jacob
  {\O}stergaard}.} \bibinfo{year}{2016}\natexlab{}.
\newblock \showarticletitle{Real-time procurement strategies of a proactive
  distribution company with aggregator-based demand response}.
\newblock \bibinfo{journal}{\emph{IEEE Transactions on Smart Grid}}
  \bibinfo{volume}{9}, \bibinfo{number}{2} (\bibinfo{year}{2016}),
  \bibinfo{pages}{766--776}.
\newblock


\end{thebibliography}

\newpage
\appendix

\section{Nomenclature}

\begin{table}[h!]
	\centering
		\begin{tabular}{ll}
			\hline
			$\mathcal{T}$ & Time steps in optimization period \\
			$\mathcal{V}$ &  Nodes in gas and electricity networks \\
			$\mathcal{V}_j$ &  Nodes connected to node $j$ \\
			$S_j$ & Bids from gas and electricity suppliers at node $j$ \\
			$\Phi^u_j$ &  Time periods with distinct start-up costs of supplier $j$ \\
			$\Phi^{u,\text{init}}_{j}$ & Initial up- or down-time periods of supplier $j$ \\
			$\mathcal{A}$ & Connections in gas network\\
			$\mathcal{A}_v$ & Control valves in gas network departing from node $j$\\
			$\mathcal{A}_v$ & Compressors in gas network\\
			$\mathcal{A}_p$ & Pipelines in gas network\\
			$a_t$ & Tail of connection $a \in \mathcal{A}$\\
			$a_h$ & Hail of connection $a \in \mathcal{A}$\\
			$\underline{\Phi}^\uparrow_{jt}$ & Minimum up-time of supplier $j$ at time $t$ (h) \\
			$\underline{\Phi}^\downarrow_{jt}$ & Minimum down-time of supplier $j$ at time $t$  (h) \\
			$u_{j0}^{\text{init}}$ & Initial on/off state of supplier $j$ \\
			$\underline{\pi}_{j},\overline{\pi}_{j}$ & Pressure bounds at node $j$ at time $t$ (MWh) \\
			$\underline{\phi}_{a},\overline{\phi}_{a}$ & Flow bounds of connection $a$ at time $t$ (MWh) \\
			$\phi_{a}^k$ & Discretized values of flow in pipeline $a$ (MWh) \\
			$\underline{c}^c_{a},\overline{c}^c_{a}$ & Ratio bounds of compressor $a$ \\			
			$\underline{c}^v_{a},\overline{c}^v_{a}$ & Ratio bounds of valve $a$ \\
			$W_{a},\overline{c}^v_{a}$ & Ratio bounds of pipeline $a$ \\
			$\alpha_j$ & Price parameter of GFPP at node $j$ \\
			$c^{uc}_{jb}$ & No load cost of supplier at node $j$ (Wh) \\
			$c^e_{jb}$ & Price of electricity bid $b$ at node $j$ (Wh) \\
			$c^g_{jb}$ & Price of electricity bid $b$ at node $j$ (Wh) \\
			$d^e_{jt}$ &  Electricity demand at node $j$ and time $t$ (Wh) \\
			$d^g_{jt}$ &  Gas demand at node $j$ and time $t$ (Wh) \\
			$\overline{s}^e_{jb}$ & Quantity bound of electricity bid $b$ at node $j$  \\
			$\overline{s}^g_{jb}$ & Quantity bound of gas bid $b$ at node $j$  \\
			$\underline{s}^g_{j},\overline{s}^g_{j}$ & Gas production bounds at node $j$  \\
			$B_{jk}$ &  Susceptance of line connecting nodes $j$ and $k$ (S) \\
			$\overline{f}_{jk}$ & Maximum flow in line connecting nodes $j$ and $k$ (Wh) \\
			\hline
	\end{tabular}
	\caption{Sets of indexes and parameters in electricity and gas systems}
	\label{table1}
\end{table}

\begin{table}[h!]
	\centering
	\begin{tabular}{ll}
		\hline
		$u_{jbt}$ & Commitment variable of bid $b$ of unit $j$ at time $t$ \\
		$v^\uparrow_{jt}$ & Start-up variable of unit $j$ at time $t$ \\
		$v^\downarrow_{jt}$ & Shut-down variable of unit $j$ at time $t$ \\
		$r_{jt}$ & Start-up cost of unit $j$ at time $t$ \\
		$s^e_{jbt}$ & Dispatch of electricity bid $b$ of supplier $j$ at time $t$ (Wh) \\
		$s^g_{jbt}$ & Dispatch of gas bid $b$ of supplier $j$ at time $t$ (Wh) \\
		$\gamma^g_{jt}$ & Gas consumption of GFPP $j$ at time $t$ (Wh) \\
		$q_{jt}$ & Gas load shedding at node $j$ at time $t$ (Wh) \\
		$y_{jt}^g$ & Gas market-clearing price at node $j$ at time $t$ (\$/Wh) \\
		\hline
	\end{tabular}
	\caption{Primal decision variables of the UC, EM, and GM clearing problems}
	\label{table3}
\end{table}

\section{Mathematical Formulations}
\label{app_math}

\subsection*{\textbf{Gas-aware electricity UC problem}}

The gas-aware electricity UC problem described in Section \ref{section:4} can be formulated as follows:
As explained in the companion paper, the upper-level problem represents the heat UC problem, which can be formulated as follows:
\begin{subequations} \label{ul}
	\begin{alignat}{2}
	& \min \ && \sum_{t \in \mathcal{T}} \sum_{j \in \mathcal{V}} \left( c_{jt}^{uc} u_{j0t}    +  r_{jt} + \sum_{b \in S_{j}} c_{jbt}^{e} s^e_{jbt}  \right) \label{ul1} \\  
	& \text{s.t.} \ &&	r_{jt} \geq  c_{jh}^\uparrow \left(   u_{jt}^0 - \sum_{k = t - h}^{t} u_{jk}^0 \right) , \forall j \in \mathcal{V} , t \in \mathcal{T} , h \in \Phi^u_{jt} \label{ul2}  \\		
	& \quad && r_{jt} \geq 0  , \forall j \in \mathcal{V} , t \in \mathcal{T} \label{ul3} \\
	& \quad && u_{j0t} = u_{j}^{\text{init}} , \forall j \in \mathcal{V} , t \in \Phi^{u,\text{init}}_{j} \label{ul4} \\
	& \quad && \sum_{k=t-\underline{\Phi}^\uparrow_{jt} + 1}^{t} v^\uparrow_{jk} \leq u_{j0t} , \forall j \in \mathcal{V} , t \in \mathcal{T} \setminus \Phi^{u,\text{init}}_{j} \label{ul5} \\
	& \quad && \sum_{k=t-\underline{\Phi}^\downarrow_{jt} + 1}^{t} v^\uparrow_{jk} \leq 1 - u_{j0(t-\underline{\Phi}^\downarrow_{jt})}    , \forall j \in \mathcal{V} ,  t \in \mathcal{T} \setminus \Phi^{u,\text{init}}_{j} \label{ul6} \\
	& \quad && v^\uparrow_{jt}  - v^\downarrow_{jt}  = u_{j0t}  - u_{j0(t-1)} , \forall j \in \mathcal{V}, t \in \mathcal{T} \label{ul7} \\
	& \quad && u_{jbt} \leq u_{j(b-1)t}  , \forall j \in \mathcal{V} , t \in \mathcal{T}, b \in S_{j} \label{ul8} \\
	& \quad && v^\uparrow_{jt}  , v^\downarrow_{jt}  , u_{j0t} \in \{0,1\} , \forall  j \in \mathcal{V} , t \in \mathcal{T} \label{ul8.2} \\ 
	& \quad && u_{jbt} \in \{0,1\} , \forall j \in \mathcal{V} , t \in \mathcal{T}, b \in S_{j}  \label{ul8.3} \\
	& \quad && \left( c^\text{e}_{j(b-1)} - M_j \right) \left( u_{j(b-1)t} - u_{jbt} \right) \geq 2\alpha_jy^g_{jt} - M_j  \nonumber \\
	& \quad &&    ,  \forall j \in \mathcal{V}^{GFPP}, t \in \mathcal{T}, b \in S_{j} \label{ul10} \\ 
	& \quad && \{s^e_{jbt},y^g_{jt}\} \in \text{ primal and dual sol. of }\eqref{ll}  \label{follower} 	
	\end{alignat}
\end{subequations}
Equations \eqref{ul2} and \eqref{ul3} model the start-up cost depending on the time the units have been offline. Indeed, the expression $\left( u_{jt}^0 - \sum_{k = t - h}^{t} u_{jk}^0 \right)$ is one when unit $j$ becomes online after it has been turned off for $h$ time periods. Equation \eqref{ul4} fixes the initial minimum up- and down-time of the units. Equations \eqref{ul5} and \eqref{ul6} enforce the minimum up- and down-time, respectively. Equations \eqref{ul7} and \eqref{ul8} state the relationship between the binary variables for the on-off, start-up, and shut-down statuses of each unit. Equation \eqref{ul8} ensures that a bid is selected only if the previous bid has been selected. Furthermore, Equation \eqref{ul10} represents the bid-validity constraints, which ensure that the price of the last selected bid of each GFPP is greater than their marginal electricity production cost. Finally, Equation \eqref{follower} represents the feedback from the primal and dual solutions of the middle- and lower-level problems.

\subsection*{\textbf{EM clearing problem}}

\begin{subequations} \label{ll}
	\begin{alignat}{2}
	& \min \ && \sum_{t \in \mathcal{T}} \sum_{j \in \mathcal{V}} \sum_{b \in S_j} c_{jb}^e s^e_{jbt} \label{ll1} \\
	& \text{s.t.} \ &&  \sum_{b \in S_j} s^e_{jbt} \nonumber \\
	& \quad &&   = d^e_{jt} + \sum_{k \in \mathcal{V}_j} B_{kj} (\theta_{mt} - \theta_{nt} ) , \forall j \in \mathcal{IV},t \in \mathcal{T}  \label{ll2} \\
	& \quad && B_{kj}\left(\theta_{kt} - \theta_{jt}\right) \leq \overline{f}_{kj} , \forall j \in \mathcal{V}, k \in \mathcal{V}_j, t \in \mathcal{T}  \label{ll3} \\
	& \quad && \overline{s}_{jb}^e u_{j(b+1)t} \leq  s_{jbt}^e \leq  \overline{s}^e_{j(b-1)} u_{jbt}, \forall j \in \mathcal{IV}, b \in S_j, t \in \mathcal{T}, \label{ll4} 
	\end{alignat}
\end{subequations}
Equation \eqref{ll2} represents the power balance equation at each node of the power system. Equation \eqref{ll3} bounds the power flow between two nodes. The power transmission network is modeled using the standard linearized DC power flow. Finally, Equation \eqref{ll4} represents the upper and lower bound of each selected bid.

\subsection*{\textbf{GM clearing problem}}

The GM clearing problem can be formulated as the following optimization problem:
\begin{subequations} \label{eq:gas}
\begin{align}
\mathcal{O}^*\left( d_{jt}^g \right) = 
\min ~~& \sum_{t \in \mathcal{T}} \sum_{j \in \mathcal{V}} 
\left(c_{jb}^g \, \bm{s_{jbt}^g} \ + \ 
\kappa_j \, \bm{q_{jt}} \right) \\
{s.t.~~} & 
\sum_{b \in S_j} \bm{s_{jbt}^g} -  d_{jt}^g + \bm{q_{jt}}  - \gamma^g_{jt} \nonumber \\
& = \sum_{a \in \mathcal{A}:j=a_h } \bm{\phi_{at}} - \sum_{a \in \mathcal{A}:j=a_t } \bm{\phi_{at}} 
\:\: \forall j \in \mathcal{V}, t \in \mathcal{T} \label{eq:gas_1} \\
& 0 \leq \bm{q_{jt}} \leq d_{jt}^g
\:\: \forall j \in \mathcal{V}, t \in \mathcal{T} \label{eq:gas_2} \\
& \underline{s}_j^g \leq  \sum_{b \in S_j} \bm{s_{jbt}^g}  \leq \overline{s}_j^g
\:\: \forall j \in \mathcal{V}, t \in \mathcal{T} \label{eq:gas_3}\\
& 0 \leq \bm{s_{jbt}^g} \leq \overline{s}_{jb}^g
\:\: \forall j \in \mathcal{V}, b \in S_j , t \in \mathcal{T} \label{eq:gas_4}\\ 
& \underline{c}_a^c \bm{\pi_{a_ht}} \leq \bm{\pi_{a_tt}} \leq \overline{c}_a^c \bm{ \pi_{a_ht} }
\:\: \forall a \in \mathcal{A}_c, t \in \mathcal{T} \label{eq:gas_5}\\
& \underline{c}_a^v \bm{\pi_{a_ht}} \leq \bm{\pi_{a_tt}} \leq \overline{c}_a^v \bm{\pi_{a_ht} }
\:\: \forall a \in \mathcal{A}_v, t \in \mathcal{T} \label{eq:gas_6}\\
& \bm{\pi_{a_ht}} - \bm{\pi_{a_tt}} \geq 2 W_a \phi^k_{a} \bm{\phi_{at}} - \left( \phi^k_{a} \right)^2 \nonumber \\ 
& \:\: \forall a \in \mathcal{A}_p, k \in \{1,...,K\} t \in \mathcal{T} \label{eq:gas_7} \\
& \underline{\pi}_j \leq \bm{\pi_{jt}} \leq \overline{\pi}_j
\:\: \forall j \in \mathcal{V}, t \in \mathcal{T} \label{eq:gas_8}\\
& \underline{\phi}_a \leq \bm{\phi_{at}} \leq \overline{\phi}_a
\:\: \forall a \in \mathcal{A}_p, t \in \mathcal{T} \label{eq:gas_9}
\end{align}
\end{subequations}

The post-processing steps can be formulated as follows:
%
%

\subsection*{\textbf{dual GM clearing problem}}

For a given value of the gas demands $\bm{\hat{d}_{jt}^g}$, using the equivalent KKT conditions of Problem \eqref{eq:gas}, the lower-level problems in the post-processing algorithms can be replaced by the following set of equations:
\begin{subequations} \label{eq:gas_KKT}
	\begin{align}
	& 
	\sum_{b \in S_j} \bm{s_{jbt}^g} -  \bm{\hat{d}_{jt}^g} + \bm{q_{jt}}  - \gamma_{jt} \nonumber \\ 
& = \sum_{a \in \mathcal{H}^p_j } \bm{\phi_{at}} - \sum_{a \in \mathcal{T}^p_j} \bm{\phi_{at}} \:\: \forall j \in \mathcal{V}, t \in \mathcal{T} \label{eq:gas_KKT_1} \\
	& 0 \leq \underline{\gamma}^q_{jt} \perp \left( - \bm{q_{jt}} \right) \leq 0 \:\: \forall j \in \mathcal{V}, t \in \mathcal{T}  \label{eq:gas_KKT_2}  \\
	& 0 \leq \overline{\gamma}^q_{jt} \perp \left( \bm{q_{jt}} - \bm{\hat{d}_{jt}^g} \right) \leq 0 \:\: \forall j \in \mathcal{V}, t \in \mathcal{T}  \label{eq:gas_KKT_3} \\
	& 0 \leq \underline{\gamma}^s_{jt} \perp \left( \underline{s}_j^g -  \sum_{b \in S_j} \bm{s_{jbt}^g} \right)  \leq 0 \:\: \forall j \in \mathcal{V}, t \in \mathcal{T} \label{eq:gas_KKT_4} \\
	& 0 \leq \overline{\gamma}^s_{jt} \perp \left( \sum_{b \in S_j} \bm{s_{jbt}^g} - \overline{s}_j^g \right)  \leq 0   \:\: \forall j \in \mathcal{V}, t \in \mathcal{T} \label{eq:gas_KKT_5} \\
	& 0 \leq \underline{\gamma}^s_{jbt} \perp \left( - \bm{s_{jbt}^g} \right) \leq 0
	\:\: \forall j \in \mathcal{V}, b \in S_j , t \in \mathcal{T} \label{eq:gas_KKT_6}  \\ 
	& 0 \leq \overline{\gamma}^s_{jbt} \perp \left( \bm{s_{jbt}^g} - \overline{s}_s^g \right) \leq 0   \:\: \forall j \in \mathcal{V}, b \in S_j , t \in \mathcal{T} \label{eq:gas_KKT_7} \\ 
	& 0 \leq \underline{\mu}^\pi_{at}  \perp \left( \underline{c}_a^c \bm{\pi_{a_ht}} - \bm{\pi_{a_tt}} \right) \leq 0	\:\: \forall a \in \mathcal{A}_c, t \in \mathcal{T} \label{eq:gas_KKT_8} \\
	& 0 \leq \overline{\mu}^\pi_{at}  \perp \left( \bm{\pi_{a_tt}} - \overline{c}_a^c \bm{\pi_{a_ht}}  \right) \leq 0	\:\: \forall a \in \mathcal{A}_c, t \in \mathcal{T} \label{eq:gas_KKT_9} \\
	& 0 \leq \underline{\mu}^\pi_{at}  \perp \left( \underline{c}_a^v \bm{\pi_{a_ht}} - \bm{\pi_{a_tt}} \right) \leq 0	\:\: \forall a \in \mathcal{A}_v, t \in \mathcal{T} \label{eq:gas_KKT_10} \\
	& 0 \leq \overline{\mu}^\pi_{at}  \perp \left( \bm{\pi_{a_tt}} - \overline{c}_a^v \bm{\pi_{a_ht}}  \right) \leq 0	\:\: \forall a \in \mathcal{A}_v, t \in \mathcal{T} \label{eq:gas_KKT_11} \\
	&  0 \leq \underline{\mu}^\phi_{a,t,k} \perp \left( 2 W_a \phi^k_{a} \bm{\phi_{at}} - \left( \phi^k_{a} \right)^2 - \bm{\pi_{a_ht}} + \bm{\pi_{a_tt}} \right) \leq 0 \nonumber \\ 
    & \:\: \forall a \in \mathcal{A}_p, k \in \{1,...,K\} t \in \mathcal{T}  \label{eq:gas_KKT_12} \\
	& 0 \leq \underline{\gamma}^\pi_{jt} \perp \left( \underline{\pi}_j - \bm{\pi_{jt}} \right) \leq 0	\:\: \forall j \in \mathcal{V}, t \in \mathcal{T} \label{eq:gas_KKT_13} \\
	& 0 \leq \overline{\gamma}^\pi_{jt} \perp \left( \bm{\pi_{jt}} - \overline{\pi}_j  \right) \leq 0  \:\: \forall j \in \mathcal{V}, t \in \mathcal{T} \label{eq:gas_KKT_14} \\
	& 0 \leq \underline{\gamma}^\phi_{at} \perp \left( \underline{\phi}_a - \bm{\phi_{at}} \right) \leq 0	\:\: \forall a \in \mathcal{A}_p, t \in \mathcal{T} \label{eq:gas_KKT_15} \\
	& 0 \leq \overline{\gamma}^\phi_{at} \perp \left( \bm{\phi_{at}} - \overline{\phi}_a  \right) \leq  0 \:\: \forall a \in \mathcal{A}_p, t \in \mathcal{T} \label{eq:gas_KKT_16} \\
	& c_{jb} - y^g_{jt} - \underline{\gamma}^s_{jt} + \overline{\gamma}^s_{jt} - \underline{\gamma}^s_{jbt} + \overline{\gamma}^s_{jbt} = 0  \:\: \forall j \in \mathcal{V}, b \in S_j , t \in \mathcal{T}  \label{eq:gas_KKT_17} \\
	& \kappa_j - y_{jt}^g - \underline{\gamma}^q_{jt} + \overline{\gamma}^q_{jt} = 0  \:\: \forall j \in \mathcal{V}, t \in \mathcal{T} \label{eq:gas_KKT_18} \\
	& - \underline{\gamma}_{jt}^\pi + \overline{\gamma}^\pi_{jt} + \sum_{a \in \mathcal{T}^c_j } \left( - \underline{\mu}_{at}^\pi + \overline{\mu}_{at}^\pi \right) +  \sum_{a \in \mathcal{H}^c_j } \left( \underline{c}^c_a \underline{\mu}_{at}^\pi  - \overline{c}^c_a \overline{\mu}_{at}^\pi \right) \nonumber \\ 
& + \sum_{a \in \mathcal{T}^v_j } \left( - \underline{\mu}_{at}^\pi + \overline{\mu}_{at}^\pi \right)  +   \sum_{a \in \mathcal{H}^v_j } \left( \underline{c}^v_a \underline{\mu}_{at}^\pi  - \overline{c}^v_a \overline{\mu}_{at}^\pi \right)	\nonumber \\
	& + \sum_{k=1}^K \left( \sum_{ a \in \mathcal{T}^p_j } \underline{\mu}_{a,t,k}^\phi -  \sum_{a \in \mathcal{H}^p_j } \underline{\mu}_{a,t,k}^\phi  \right) = 0  \:\: \forall j \in \mathcal{V},  t \in \mathcal{T} \label{eq:gas_KKT_19} \\
	& y^g_{a_ht} - y^g_{a_tt} + \sum_{k=1}^{K} 2 W_a \phi_a^k \underline{\mu}_{a,t,k} - \underline{\gamma}^\phi_{at} +  \overline{\gamma}^\phi_{at} = 0 \nonumber \\
	& \:\: \forall a \in \mathcal{A}_p,  t \in \mathcal{T} \label{eq:gas_KKT_20}
	\end{align}
\end{subequations}
Equations \eqref{eq:gas_KKT_1}-\eqref{eq:gas_KKT_16} represent in a compact form the primal and dual constraints, and the complementarity conditions of Problem \eqref{eq:gas}. Equations \eqref{eq:gas_KKT_17}-\eqref{eq:gas_KKT_20} represent the stationarity conditions of Problem \eqref{eq:gas}. Furthermore, the bilinear complementarity conditions in \eqref{eq:gas_KKT_2}-\eqref{eq:gas_KKT_16} can be linearized, by using disjunctive constraints. Therefore the post-processing problems can be recast as mixed integer second order cone programs (MISOCP).

\end{document}